\documentclass[article]{amsart}

\usepackage{amsmath,amssymb}
\usepackage{graphicx}

\def\ee{\mathrm{e}}

\def\ii{{\sqrt{-1}}}

\def\mat{{1}}
\def\inf{{0}}
\def\Inf{{\rm{inf}}}
\def\Sup{{\rm{sup}}}

\def\Vol{{\mathrm{vf}}}
\def\Area{{\mathrm{\ell}}}
\def\supp{{\mathrm{supp}}}
\def\rght{{\mathrm{r}}}
\def\lft{{\mathrm{\ell}}}
\def\up{{\mathrm{u}}}
\def\dwn{{\mathrm{d}}}

\def\OS{{\mathrm{+}}}
\def\VS{{\mathrm{-}}}
\def\CS{{\mathrm{\pm}}}
\def\Path{{\mathrm{Path}}}
\def\Reg{{\CC}}
\def\Box{{B}}

\def\rad{{r}}

\def\NN{{\mathbb{N}}}
\def\RR{{\mathbb{R}}}
\def\CC{{\mathbb{C}}}
\def\cC{{\mathcal{C}}}
\def\cD{{\mathcal{D}}}
\def\cE{{\mathcal{E}}}
\def\cW{{\mathcal{W}}}
\def\cH{{\mathcal{H}}}
\def\PP{{P\mathbb{C}}}
\def\cond{{\gamma}}

\def\cQ{\mathcal{Q}}
\def\cS{\mathcal{S}}
\def\cB{\mathfrak{B}}
\def\fs{\mathfrak{s}}

\def\cO{\mathcal{O}}
\def\hcQ{{\hat{\mathcal{Q}}}}
\def\Xconf{{X}}
\def\hXconf{{\hat{X}}}

\def\dens{{\lambda}}
\def\vG{{v}}
\def\uG{{u}}

\def\pc{{p_c}}
\def\pcC{{p_c}}
\def\pcS{{p_c^*}}
\def\lc{{\lambda_c}}
\def\lcC{{\lambda_c}}
\def\lcS{{\lambda_c^*}}
\def\talpha{{t}}
\def\talphaC{{t^\OS}}
\def\talhpaS{{t^\VS}}

\def\flow{{\mathrm{flow}}}
\newcommand\Exp[1]{{{\mathbf{E}_{#1}}}}
\newcommand\Prob[1]{{{\mathbf{P}_{#1}}}}

\newcommand\hProb[1]{{{\hat{\mathbf{P}}_{#1}}}}
\newcommand\uC[1]{{u^\OS_{#1}}}

\newcommand\uS[1]{{u^\VS_{#1}}}
\newcommand\uCS[1]{{u^\CS_{#1}}}

\newcommand\gammaC[1]{{\gamma_\OS^{#1}}}
\newcommand\gammaS[1]{{\gamma_\VS^{#1}}}
\newcommand\gammaCS[1]{{\gamma_\CS^{#1}}}
\newcommand\GammaC[1]{{\Gamma^\OS_{#1}}}
\newcommand\GammaS[1]{{\Gamma^\VS_{#1}}}
\newcommand\GammaCS[1]{{\Gamma^\CS_{#1}}}
\newcommand{\Xconfi}[2]{{X^{(#1)}_{#2}}}
\newcommand\COS[1]{{C^\OS_{#1}}}
\newcommand\CVS[1]{{C^\VS_{#1}}}

\newcommand{\by}[1]{{\textrm{#1},}}
\newcommand{\jour}[1]{{\textit{#1}}}
\newcommand{\Paper}[1]{{\textrm{#1},}}
\newcommand{\vol}[1]{{\textbf{#1}}}
\newcommand{\yr}[1]{{(\textrm{#1})}}
\newcommand{\pages}[1]{{\textrm{#1}}}
\newcommand{\Book}[1]{{\textit{#1},}}
\newcommand{\publ}[1]{{(\textrm{#1},}}
\newcommand{\publaddr}[1]{{(\textrm{#1},}}
\newcommand{\byr}[1]{{\textrm{#1})}}

\newtheorem{definition}{Definition}[section]

\newtheorem{theorem}{Theorem}[section]
\newtheorem{proposition}{Proposition}[section]
\newtheorem{corollary}{Corollary}[section]
\newtheorem{remark}{Remark}[section]
\newtheorem{lemma}{Lemma}[section]

\begin{document}

\markboth{S. Matsutani and Y. Shimosako}
{On Conductivity in a continuum percolation model}

%
%

\title{ON HOMOGENIZED CONDUCTIVITY AND FRACTAL STRUCTURE 
IN A HIGH CONTRAST CONTINUUM PERCOLATION MODEL}

\author{SHIGEKI MATSUTANI}

\address{Analysis technology development center,\\
Canon Inc. 3-30-2, Shimomaruko, Ohta-ku, Tokyo 146-8501, Japan\\
matsutani.shigeki@canon.co.jp}

\author{YOSHIYUKI SHIMOSAKO}

\address{Analysis technology development center,\\
Canon Inc. 3-30-2, Shimomaruko, Ohta-ku, Tokyo 146-8501, Japan\\
shimosako.yoshiyuki@canon.co.jp}

\maketitle


\begin{abstract}
In the previous article (S. Matsutani and Y. Shimosako and Y. Wang,
    Physica A \vol{391} \yr{2012} \pages{5802-5809})
we numerically investigated an electric potential problem
with high contrast local conductivities ($\gamma_0$ and $\gamma_1$,
$0<\gamma_0 \ll \gamma_1$)  for
a two-dimensional continuum percolation model (CPM).
As numerical results, we showed there that 
the equipotential curves exhibit the fractal structure around
the threshold $p_c$
and gave an approximated curve representing a relation between
the homogenized conductivity and the volume fraction $p$ over
 $[p_c,1]$.
In this article, using the duality of the conductivities 
and the quasi-harmonic properties, we re-investigate these topics
to improve these results.
We show that at $\gamma_0\to0$, the quasi-harmonic potential problem 
in CPM is quasiconformally equivalent to a random slit problem,
which leads us to an observation between the conformal property and 
the fractal structure at the threshold.
Further we extend the domain $[p_c,1]$ of the approximated curve to 
$[0,1]$ based on the these results,
which is partially generalized to three dimensional case.
These curves represent well the numerical results of the
conductivities.
\end{abstract}



\section{Introduction}

In the series of articles \cite{MSW,MSW2,MSW3},
we numerically investigated the electric potential problem
with high contrast local conductivities ($\gamma_0$ and $\gamma_1$,
$0<\gamma_0 \ll \gamma_1$)  for
continuum percolation models (CPMs)
in order to reveal the electrical properties 
of a real material consisting of conductive nano-particles in an insulator.
We solved generalized Laplace equations with mixed boundary 
conditions using the finite difference method.
In Ref.~\cite{MSW3}, we investigated two-dimensional
case and numerically showed that 
the equipotential curves exhibit the fractal structure around
the threshold $p_c$
and gave an approximated curve representing a relation between
the homogenized conductivity and the volume fraction $p$ over
$p \in [p_c,1]$,
  which we call {\it{conductivity curve}}.
The fractal structure in Ref.~\cite{MSW3}
shows the electrical properties of the
system, i.e., there appear quasi-equipotential
clusters there, in which the potential distribution
is constant or nearly constant.
On the other hands, due to the ultra high conductivity of an insulator,
we have a non-vanishing homogenized conductivity even 
for a smaller volume-fraction $p$ than the threshold $p_c$.
It is also important to determine 
the dependence of the homogenized conductivity on
volume fraction $p$ for every  $p \in [0,1]$.

In this article,
we go on to investigate the electric potential distribution on
a two-dimensional continuum percolation model (CPM) \cite{MR}
with a high contrast local conductivity
to improve the results in Ref.~\cite{MSW3} 
from more mathematical viewpoints.
We consider the simplest Boolean model of CPM,
in which disks of unit radius are centered at the 
points of the Poisson point process
with density $\dens$ \cite{MR,SKM}.
In order to consider the duality between the occupancy 
state (OS) and the vacancy state (VS),
we handle both types in this article;
the OS-type is set so that occupied regions 
have the local conductivity $\cond_\mat=1$ and the vacant regions
have infinitesimal one $\cond_\inf$, 
whereas in the VS-type,  the occupied regions
have $\cond_\inf$ and the vacant regions have $\cond_\mat$.
VS-type is sometimes called the Swiss cheese model 
in physics literature if $\cond_\inf =0$.

Using the duality of two dimensional case,
the purposes of this article are  to find the geometric 
properties of the potential distribution, especially
the fractal structure of the equipotential curves around the threshold,
and to extend the domain $[p_c,1]$ of conductivity curve
for the homogenized conductivity to $[0,1]$,
which leads us a novel parameterization of the conductivity
curve even of three dimensional case.

The homogenized conductivity in CPM
has been rigorously studied well and has a long history as
recently Kontogiannis \cite{Kon} gave a nicer review of these studies
and an extension. The studies from more practical viewpoints 
based on the numerical results are in Refs.~\cite{J} and \cite{WJ}. 
In this article, based on these rigorous results and numerical computations,
we give an improvement of the previous arguments 
 in Ref.~\cite{MSW3}.
Thus the conventions and notations in this article completely differ from
the previous ones \cite{MSW,MSW2,MSW3} since they were written
in the framework of physics.
In general, the homogenization is assumed a periodicity of its associated
system.  Instead of the periodicity, we use the ergodicity of CPM
to homogenize the conductivity following Refs.~\cite{BMW} and 
\cite{DM}
because the Boolean model has the ergodic properties \cite{DM,Kon}.

Due to the duality of OS- and VS-types, 
it is well-known that our second order partial
differential equations (PDEs) have the quasiconformal
properties of the PDEs as the quasi-harmonic system \cite{AIM}.
Using the quasiconformal property
and the Keller-Dykhne reciprocity law \cite{Dyk,Kell} of the conductivity,
the equipotential curves are represented mathematically. 
We showed that in Theorem \ref{mapQS} 
for the infinitesimal limit of $\cond_\inf$, 
the OS-type potential distribution in CPM forms a quasiconformal map 
from the domain to $[0,1]\times[0,1]$ with random slits.
It implies that at $\gamma_0\to0$, the quasi-harmonic potential problem 
in CPM is quasiconformally equivalent to a random slit problem.
Since one of our purposes is to reveal the fractal structure of the
equipotential curves  \cite{MSW3},
we give Theorem \ref{thm:eqicurve} and using it, show
an observation in
Remark \ref{rmk:obs} that the equipotential curves 
at $p\nearrow p_c$ are conformally mapped to the family of the lines
with infinite length under some assumptions.
Further the conformal map gives a relation to a Riemann sphere $\PP^1$.

The other purpose is to find the properties
of the conductivity under the threshold. 
In order to use the duality,
we give our numerical computations of the conductivity 
of VS-type using the Monte-Carlo method and finite difference method
as we did for OS-type in Ref.~\cite{MSW3}.
Then based upon these results,
we give  novel approximation formulae of the conductivity curves.
In other words, using the duality, we give approximation formulae
of the conductivity curves in Remark \ref{rmk:R_cond2} and
Figure 8 (a). Figure 8 (a) shows that
the approximation formulae represent well the computational results
of the homogenized conductivities.

It leads us another approximation formula of
the conductivity curve of three-dimensional case as in
Remark \ref{rmk:R_cond3} and Figure 8 (b).
In Refs.~\cite{MSW} and \cite{MSW2},
we have numerically studied the homogenized conductivity
on the CPM with OS-type local conductivity mainly in three dimensional
case by solving the  generalized Laplace equation with
a certain Dirichlet-Neumann boundary conditions.
The proposed formula also
approximates well the conductivity curve of the random
spheres in Remark \ref{rmk:R_cond3} as in Figure 8 (b).
It means that in terms of the novel formula, we can parameterize
the conductivity properties discussed in 
Refs.~\cite{MSW} and \cite{MSW2}.

\bigskip

Contents in this article are as follows.
Sec.2 gives our model and mathematical preliminaries.
Sec.~2.1 is an review on our Boolean model of CPM based upon the Poisson point
process, in which we introduce OS- and VS- types,
and
Sec.~2.2 gives our PDEs or the generalized Laplace equations on
the CPMs.
Sec.~3 also shows the well-established results homogenization of the
conductivity associated with our PDEs.
In Sec.~4,
after we review the quasiconformal properties of the PDEs as
the quasi-harmonic system, 
we show that in the infinitesimal limit of $\gamma_0$,
the potential distribution is quasiconformally equivalent to
a configuration of a random slit model in Theorem \ref{mapQS},
which is the first our main result.
We also consider a conformal structure as a special case of the
quasiconformal structures.
The conformal structure leads us  Theorem \ref{thm:eqicurve} and
a novel observation on the fractal
structure in Remark \ref{rmk:obs} that the equipotential curves 
which is the second our main result.
In Sec.~5 we give our numerical computations using the Monte-Carlo
method. Based upon these results, in Sec.~6
we give  novel approximation formulae of the conductivity curves
in Remark \ref{rmk:R_cond2} and their extension to three-dimensional
case  in Remark \ref{rmk:R_cond3}, which are
the third our main result.
Figure 8 shows that they represent the computational results well.

\bigskip

\section{Models and Preliminary}


\subsection{Model, Definitions and Notations}

Let $\hcQ(\Reg)$ be the set of all countable subsets $\hXconf$
of the region $\Reg\equiv \RR^2$ satisfying $N_K(\hXconf) < \infty$
for every compact subset $K \subset \Reg$,
where
$N_K(\hXconf)$ is the number of the points of $\hXconf\cap K$.
Thus $\hcQ(\Reg)$ has the non-negative valued Radon measure and
is equipped with $\sigma$-field $\cB(\hcQ(\Reg))$.
Let $\Area$ be the Lebesgue measure of $\RR^2=\CC$
and $\NN_0=\{0,1,2,\ldots\}$.
Then we consider
the Poisson point process
$(\hcQ(\Reg), \cB(\hcQ(\Reg)), \hProb{\lambda})$,
i.e.,
for any disjoint region $\{A_1, A_2, \cdots, A_m\} \subset \cB(\Reg)$ 
such that $\Area(A_i) < \infty$ $i = 1, 2, \ldots, m$,
$N_{A_1}(\hXconf)$, \ldots, 
$N_{A_m}(\hXconf)$ are independent
random variables on the probability space
$(\hcQ(\Reg), \cB(\hcQ(\Reg)), \hProb{\lambda})$ and for
$n \in \NN_0$,
$$
\hProb{\lambda}(N_{A_i}=n) = \frac{(\lambda \Area(A_i))^n}{n!}
 \exp( (\lambda \Area(A_i))), \quad
i = 1, 2, \ldots, m, n \in \NN_0.
$$

The simplest Boolean model of CPM
$(\cQ_r(\Reg), \cB(\cQ_r(\Reg)), \Prob{\lambda})$ 
associated with
$(\hcQ(\Reg),$ $ \cB(\hcQ(\Reg)),$ $ \hProb{\lambda})$ is given via
$\cQ_r(\Reg) := \{\overline{U_\rad(\hXconf)} \ |\ 
 \hXconf \in \hcQ(\Reg)\}$, where 
$U_r(A)$ is the $\rad$-neighborhood of $A$ for fixed $\rad>0$.
$\Exp{\lambda}$ denotes
 the expectation value of $\Prob{\lambda}$  \cite{MR,SKM}.
Since for the Borel sets $B, A \subset \Reg$, we have a natural measure
$\Area_A(B) := \Area(A\cap B)$, 
$\cQ_r(\Reg)$ is equipped with
$\cB(\cQ_r(\Reg))$ induced from $\cB(\hcQ(\Reg))$.

The following result is well-known  \cite{MR}:
By letting the distance of two points in $\Reg$ denoted by $d(x,y)$,
let $d(A)$ be $\Sup\{d(x,y) \ |\ x,y \in A\}$.
\begin{proposition}\label{prop:Roy}
In two dimensional case, for
$$
\lambda_c:=\Inf \{\lambda: \Prob{\lambda}[d(X) = \infty ]>0\},
$$
$$
\lambda_c^*:=\Sup \{\lambda: \Prob{\lambda}[d(X^c) = \infty]>0\},
$$
we have the equality and inequality
$\lambda_c = \lambda_c^* < \infty$.  
\end{proposition}

In this article, we consider a finite 
domain $\Box\equiv \Box_L:=[-L,L]^2 \subset \CC=\RR^2$
 mainly for $L\gg r$. 
As in Chap.5 in Ref.~\cite{MR},
the statistical properties over $\Box_{nL}$ is
naturally extended to $\CC=\RR^2$ by taking the large $n$ limit.
$\Box$ has the natural Borel field $\cB(\Box)$ induced from
that of $\RR^2$.
As our convention, $z \in \CC$ is expressed by the 
Cartesian coordinator system $z = x +\ii y$.
We let 
the left, right, upper and bottom edges
of the boundary $\partial \Box$ of $\Box$ 
denoted by
$\partial_{\lft} \Box$, 
$\partial_{\rght} \Box$, 
$\partial_{\up} \Box$ and 
$\partial_{\dwn} \Box$ respectively;
$\partial \Box = \bigcup \partial_{\alpha} \Box$, and
$\partial_{\up} \Box \bigcap \partial_{\lft} \Box$ is a corner of $\Box$
and so on.

For example, as in Ref.~\cite{MR}, 
we have the covered volume fraction of $X\in \cQ_r(\CC)$,
\begin{equation}
p(\lambda,r) := \Exp{\lambda}(\Vol(\Xconf)), \quad
p(\Xconf) \equiv
\Vol(\Xconf) := \lim_{n\to \infty} \Vol_{\Box_{nL}} (\Xconf),
\end{equation}
by letting $\Vol_A(A'):=\ell_A(A')/\ell(A)$ for  
$A, A' \in \cB(\CC)$ ($\ell(A) \neq 0$).
Using the volume fraction and 
Proposition \ref{prop:Roy}, 
 $\lcC=\lcS$ provides the critical volume fraction 
$p_c$, which is called percolation threshold or 
merely {\it{threshold}} in this article,
$$
p_c := 1 - \ee^{-\pi \lcC  r^2} \equiv 1 - \ee^{-\pi \lcS r^2}.
$$
Since it is known that our Boolean model 
$(\cQ_r(\Reg), \cB(\cQ_r(\Reg), \Prob{\lambda})$ is ergodic  \cite{MR,SKM},
the ergodic theorem shows the following proposition
 as Theorem 6.2 in Ref.~\cite{SKM}:
\begin{proposition} \label{prop:LNTp}
For $\Xconf_0 \in \cQ_r(\CC)$,
$$
p(\Xconf_0) 
=\lim_{n\to \infty}\Exp{\lambda}\left(\Vol_{\Box_{nL}}(\Xconf) \right),
\quad \Prob{\lambda}\mbox{-a.s.}.
$$
\end{proposition}

Since we handle the boundary problem associated with CPM, 
we mainly considered these models in the finite region $\Box_L$.
We also define 
$$
\cQ_r(\Box_L):=\{ \Xconf\cap \Box_L \ |\ \Xconf  \in \cQ_r(\Reg)\}
$$
and $(\cQ_r(\Box_L),$ $ \cB(\cQ_r(\Box_L)),$ $ \Prob{\lambda})$
by induced by the restriction map 
$\varrho_{\Box,\Box'}:\cQ_r(\Box')\to \cQ_r(\Box)$,
($\varrho_{\Box,\Box'}(\Xconf') =(\Xconf'\cap \Box)$)
for $\Box \subset \Box' \subset\CC$.

\bigskip

As mentioned above,
in $(\cQ_r(\Box), \cB(\cQ_r(\Box)), \lambda)$, the ordinary measure
is given by $\Area_U(X)= \Area(X\cap U)$ for a Borel set $U\subset \Box$
and $X\in \cQ_r(\Box)$ as the occupied region $X$. As its vacancy, 
$\Area^*_U(X)= \Area(X^c\cap U)$ is also handled where
 $X^c = \Box \setminus X$ in ordinary CPM.
In this article, we also consider the binary measures such that
for a Borel set $U \subset \Box$,
\begin{eqnarray}
    \gammaC{\Xconf}(U) &:=& \cond_\inf\ \Area_U^*(X) + \cond_\mat\ \Area_U(X),
   \nonumber\\
    \gammaS{\Xconf}(U) &:=& \cond_\inf\ \Area_U(X) + \cond_\mat\ \Area_U^*(X),
\label{eq:measure_cond}
\end{eqnarray}
where $\cond_\mat = 1$ and $\cond_\mat \gg \cond_\inf > 0$. In the limit 
$\cond_\inf \to 0$, $\gammaC{X}$ (resp. $\gammaS{X}$) 
corresponds to the measure of the occupancy (resp. vacancy).
As mentioned in Introduction,
we call the systems with $\gammaC{X}$ {\it{OS-type}} and
 $\gammaS{X}$ {\it{VS-type}} respectively.
We regard the measure $\gammaCS{}$ as 
an element of the real valued generalized functions $\cD(\Box, \RR)$ 
over $\Box$ and
we treat it as the  binary local conductivity $\gammaCS{}(z)$, i.e.,
$\gammaCS{}(z)$ is 
defined through the measure $(\cB(\CC), \gammaCS{\Xconf})$ for
every Borel set $U$ of $\CC$,
\begin{equation}
   \int_U \gammaCS{\Xconf}(dz) = \int_U \gammaCS{\Xconf}(z) d^2 z,  
\label{eq:condDist}
\end{equation}
where 
$\displaystyle{
d^2 z = dx dy = \frac{1}{2\ii} dz d\bar z
}$ is the Lebesgue measure of $\RR^2 =\CC$.
Here these $\gammaCS{\Xconf}$ are weakly first differentiable 
functions  \cite{GiTr}.

\subsection{Dirichlet-Neumann Boundary Problems on $\Box_L$}

For a configuration $X \in \cQ_r(\Box_L)$,
using the binary local conductivity $\gammaCS{\Xconf}(z)$ of
the generalized functions $\cD(\Box_L, \RR)$,
we consider second order PDEs with mixed boundary conditions.

The mixed boundary problem of the
 elliptic differential equation with measurable coefficients
is studied well in Refs.~\cite{GiTr}, \cite{Tr73} and
\cite{Tr77}.
We consider the symmetric form on the function space $L^2(\Box_L)$,
\begin{eqnarray}
    \cE_{B_L}[\gammaCS{\Xconf}; \vG, \vG']&:=&\int_{\Box_L}   \
     \gammaCS{\Xconf}(z) (\nabla \vG(z)) (\nabla \vG'(z)) d^2 z
        \nonumber\\
    &=&\int_{\Box_L}   \
(\nabla \vG(z)) (\nabla \vG'(z)) \gammaCS{\Xconf}(dz) .
\label{eq:Ws}
\end{eqnarray}
Let $\cE_{B_L}[\gammaCS{\Xconf}; \vG] := \cE_{B_L}[\gammaCS{\Xconf}; \vG, \vG]$.
For a non-negative parameter $\mu$, we define the seminorm:
$$
\|\vG\|_{\gammaCS{\Xconf},\mu,U}
:=\left\{\int_U
     \left(\gammaCS{\Xconf}(z) (\nabla \vG(z))^2 +\mu \vG(z)^2
      \right) d^2 z
\right\}^{1/2},
$$
and the Sobolev spaces,
\begin{eqnarray}
\cW^1_2(\gammaCS{\Xconf},\mu, {\Box_L})&:=&\{ \vG \in L^2({\Box_L}) \ | \ 
\|\vG\|_{\gammaCS{\Xconf},\mu,{\Box_L}} < \infty\},\nonumber\\
\cH(\gammaCS{\Xconf},\mu, {\Box_L})&:=&\{ \vG \in \cW^1_2({\Box_L}) \ | \ 
\mbox{there exists a sequence }\ \nonumber \\
& &\{\vG_m\}\subset \cC^1({\Box_L}) \mbox{ satisfying }
\|\vG_m-\vG\|_{\gammaCS{\Xconf},\mu,{\Box_L}} \to 0\}.\nonumber
\end{eqnarray}

From Proposition 1 in Ref.~\cite{Tr77}, 
$\cW^1_2(\gammaCS{\Xconf},\mu, {\Box_L})$ and
$\cH(\gammaCS{\Xconf},\mu, {\Box_L}) = \cH(\gammaCS{\Xconf}, 0, {\Box_L})$
are Hilbert spaces.
The associated PDE, which we call 
{\it{the generalized Laplace equation}}  \cite{Br,GiTr}, is
given by
\begin{equation}
\nabla \gammaCS{\Xconf}(z) \nabla \uG^\CS(z) = 0
\quad \mbox{ for }z \in {\Box_L}^o.
\label{eq:0-1}
\end{equation}
In this article, we consider the potential problems of $\gammaCS{\Xconf}$ in
(\ref{eq:0-1}) with the following boundary conditions respectively,
\begin{eqnarray}
\mbox{BC}_{B_L}^{\OS} : &\left\{
\begin{array}{cl}
\uG^\OS(z) = \uG_0^\OS,  &\mbox{ for }z \in \partial_\up {\Box_L},\\
\uG^\OS(z) = 0,  &\mbox{ for }z \in \partial_\dwn {\Box_L},\\
\frac{\partial}{\partial x}\uG^\OS(z) = 0,  &\mbox{ for }z \in 
\partial_\lft {\Box_L}\cup \partial_\rght {\Box_L},
\end{array} \right.&
\nonumber\\
\mbox{BC}_{B_L}^{\VS} : &\left\{
\begin{array}{cl}
\uG^\VS(z) = \uG_0^\VS,  &\mbox{ for }z \in \partial_\rght {\Box_L},\\
\uG^\VS(z) = 0,  &\mbox{ for }z \in \partial_\lft {\Box_L},\\
\frac{\partial}{\partial y}\uG^\VS(z) = 0, &\mbox{ for }z \in 
\partial_\up {\Box_L}\cup \partial_\dwn {\Box_L}.
\end{array} \right.&
\label{eq:0-3}
\end{eqnarray}
Here $\uCS{0}$ is a real positive constant number, $\uCS{0}>0$.
The boundary conditions are listed in Table \ref{table:BC}.
\begin{table}[htbp]
\caption{The conductivity distribution (CD) and boundary condition (BC):}
\label{table:BC}
\begin{center}
\begin{tabular}{|c|c|c|c|c|c|c|}
\hline
\multicolumn{1}{|c|}{ } &
\multicolumn{2}{|c|}{$\gammaCS{X}$}&
\multicolumn{4}{|c|}{BC} \\
\hline
      &$\Xconf$ & $\Xconf^c$ & $\partial_\up{\Box_L}$
      & $\partial_\dwn{\Box_L}$ & $\partial_\rght{\Box_L}$ & $\partial_l{\Box_L}$  \\
\hline
OS& $\cond_\mat$ & $\cond_\inf$ 
     & $\uC{}=\uC{0}$ 
     & $\uC{}=0$ 
     & $\partial_x \uC{}=0$ 
     & $\partial_x \uC{}=0$ \\
VS& $\cond_\inf$ & $\cond_\mat$ 
     & $\partial_y \uS{}=0$ 
     & $\partial_y \uS{}=0$
     & $\uS{}=\uS{0}$ 
     & $\uS{}=0$ \\
\hline
\end{tabular}
\end{center}
\end{table}

These problems are equivalent with the energy minimal problems \cite{Br},
in which we obtain the minimal functional $\cE_{B_L}^{(0)}[\gammaCS{\Xconf}]$,
$$
\cE_{B_L}^{(0)}[\gammaCS{\Xconf}]:=
\min_{\vG \in \cH({\Box_L}|\mathrm{BC}^\CS,\uG_0^\CS)}
 \cE_{B_L}[\gammaCS{\Xconf}; \vG],
$$
where
$\cH({\Box_L}|\mathrm{BC}_{B_L}^\CS,\uG_0^\CS)$ is
 the function spaces satisfying
 these boundary conditions,
\begin{eqnarray}
\cH({\Box_L}|\mathrm{BC}_{B_L}^\CS,\uG_0^\CS)&:=&\{\vG \in 
  \cH(\gammaC{\Xconf}, \mu, {\Box_L}) \ | \ 
\vG(z) \mbox{ satisfies BC}_{B_L}^{\CS}\}.
\label{eq:HBCS}
\end{eqnarray}
This $\cE_{B_L}^{(0)}[\gammaCS{\Xconf}]$ is the same as 
the {\it{capacity}}, or total conductivity in our problem \cite{GiTr}.
In other words,
for the solutions of  the generalized Laplace equations of
(\ref{eq:0-1}) with (\ref{eq:0-3}) for a configuration $\Xconf$,
we have
\begin{equation}
\cE_{B_L}^{(0)}[\gammaCS{\Xconf}]= (\uG_0^\CS)^2 \GammaCS{{\Box_L}}(\Xconf),
\label{eq:E=Gamma}
\end{equation}
where 
the total conductivities $\GammaC{{\Box_L}}(X)$
of OS-type and $\GammaS{{\Box_L}}(X)$ of VS-type are defined by
\begin{eqnarray}
\GammaC{{\Box_L}}(\Xconf)&:=&\frac{1}{\uG_0^\OS} 
\int_{\partial_\up{\Box_L}}  \gammaC{\Xconf}(z)
              \frac{\partial}{\partial y} \uC{}(z) dx,
              \nonumber \\
\GammaS{{\Box_L}}(\Xconf)&:=&\frac{1}{\uG_0^\VS} 
\int_{\partial_\rght{\Box_L}}  \gammaS{\Xconf}(z)
              \frac{\partial}{\partial x} \uS{}(z)dy,
\label{eq:defcondtotal}
\end{eqnarray}
respectively.
(\ref{eq:E=Gamma}) is obtained by Gauss-Stokes theorem.
The OS-type case is given by integrating the normal current 
$\gammaC{\Xconf} \partial \uC{}/\partial y$
along the line parallel to the $x$-axis and the VS case is
given as integral of $\gammaC{\Xconf} \partial \uC{}/\partial x$ for
 the $y$-axis.

Our motivation is illustrated in Figure \ref{fig:zoomin},
in which we show the numerical result of the equation
(\ref{eq:0-1}) with the boundary conditions (\ref{eq:0-3})
as mentioned in Sec.~5.
It shows that 1) the equipotential curves of OS-type and VS-type 
 intersect orthogonally,
2) for $\lambda < \lambda_c$, the equipotential curves of OS-type
don't penetrate into $X^\circ$ whereas
for $\lambda > \lambda_c$, the equipotential curves of VS-type
avoid $X^c$,
and 3)  the equipotential curves look very complicated.
In this article, we characterize these properties as follows.

\begin{figure}[h]
\begin{minipage}{0.45\hsize}
 \begin{center}
  \includegraphics[height=4cm]{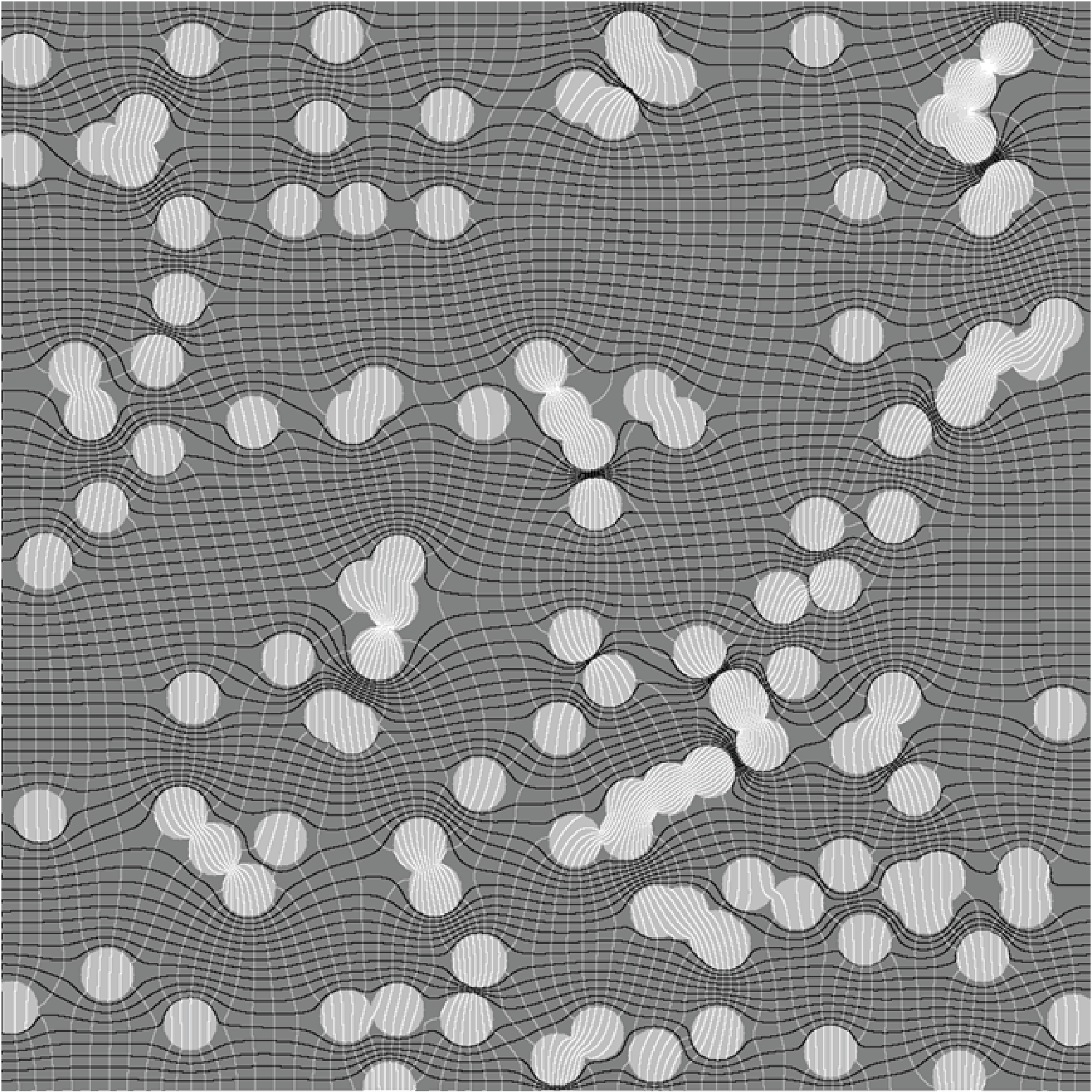} \newline
  (a)
 \end{center}
\end{minipage} 
\begin{minipage}{0.45\hsize}
 \begin{center}
  \includegraphics[height=4cm]{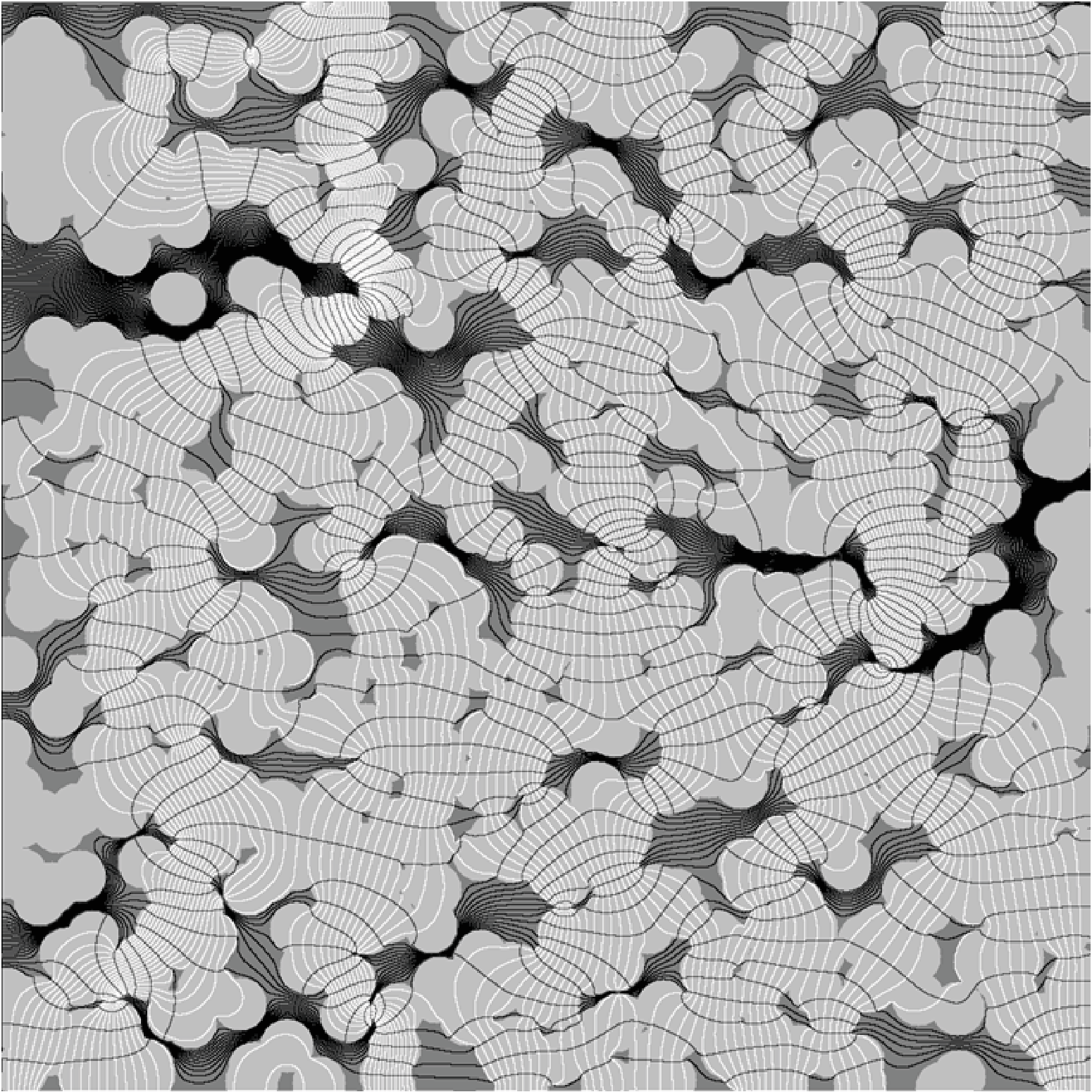} \newline
  (b)
 \end{center}
\end{minipage} \\
\begin{minipage}{0.45\hsize}
 \begin{center}
  \includegraphics[height=4cm]{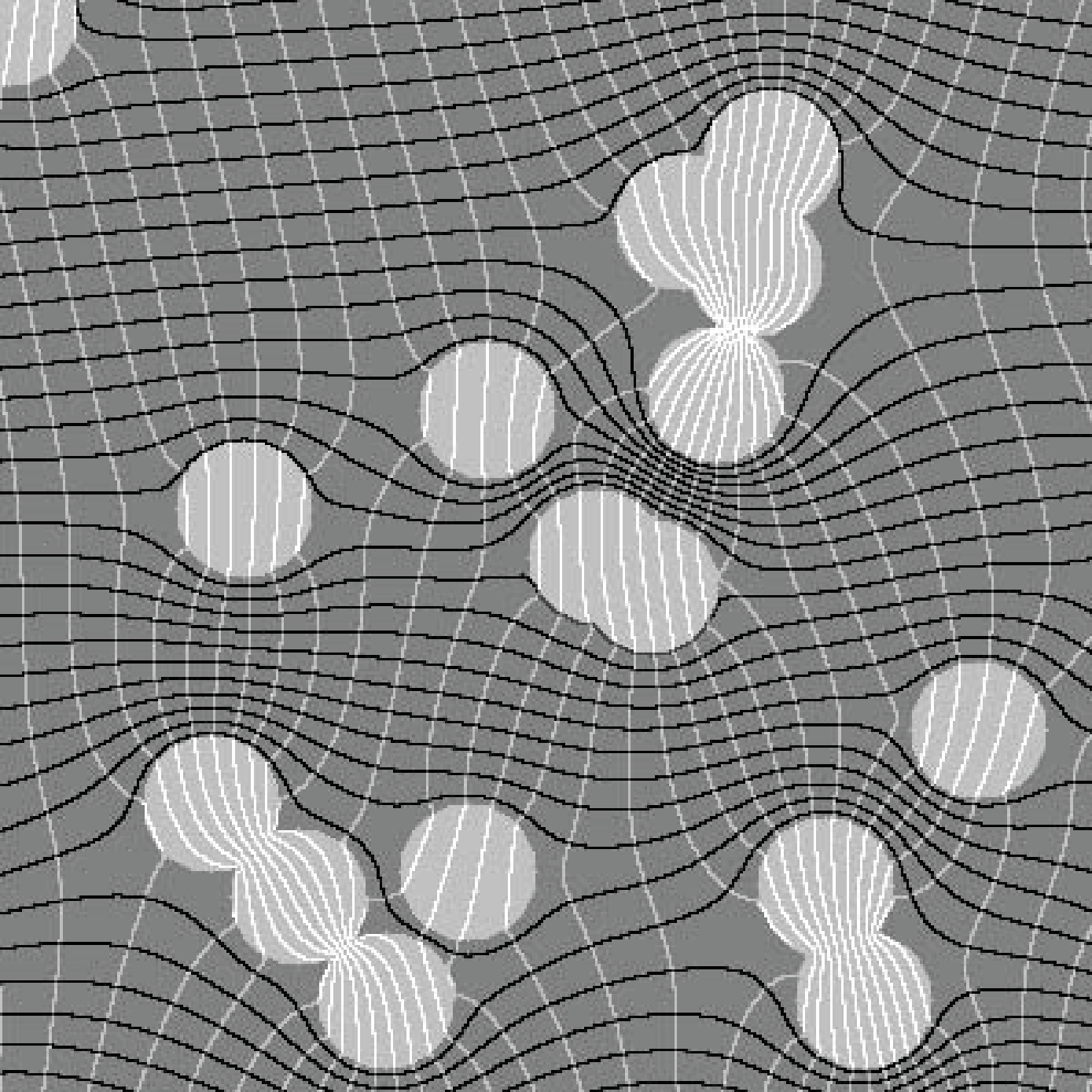} \newline
  (c)
 \end{center}
\end{minipage}
\begin{minipage}{0.45\hsize}
 \begin{center}
  \includegraphics[height=4cm]{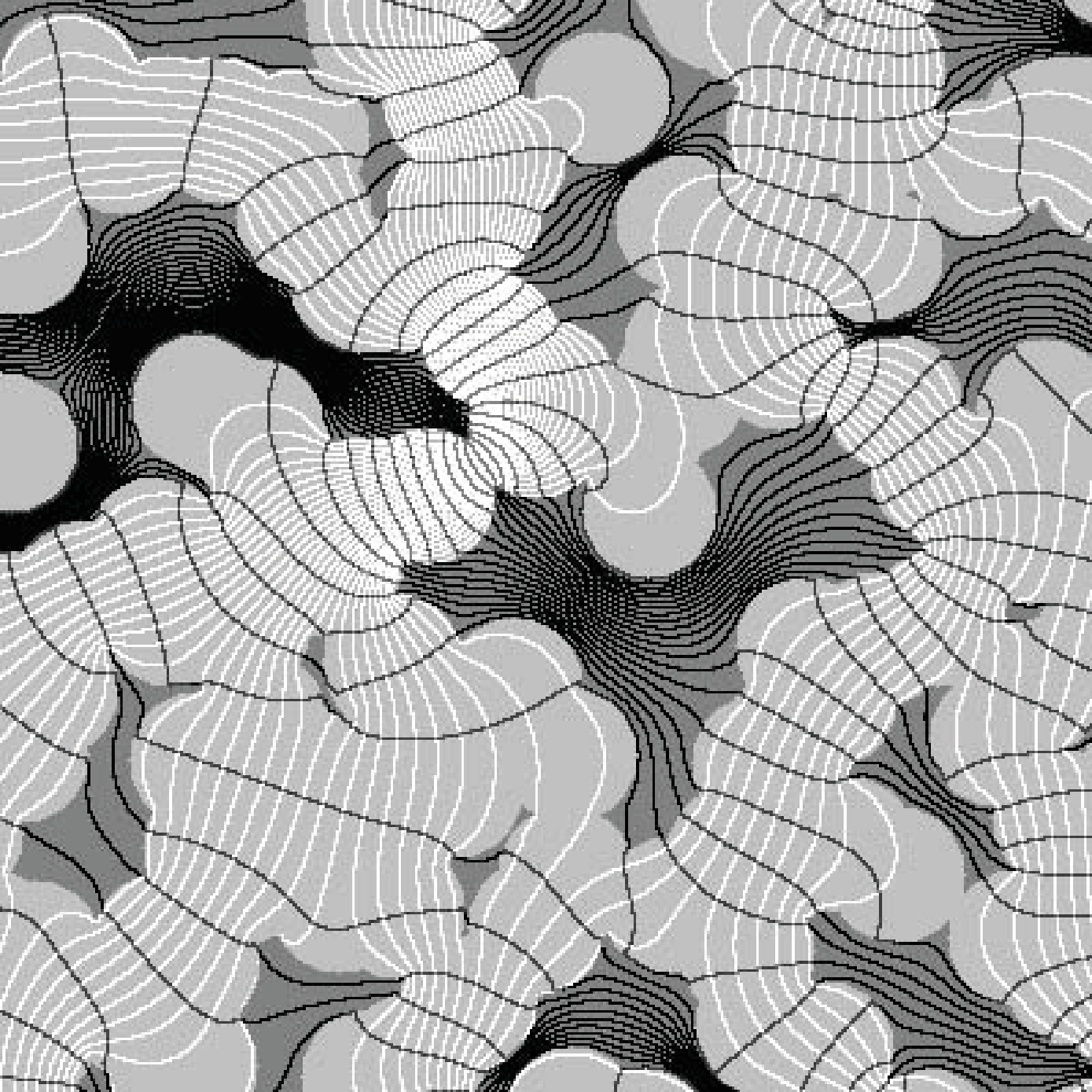} \newline
  (d)
 \end{center}
\end{minipage}
\caption{The equipotential curves of OS-type and VS-type:
The black curve represents an equipotential curve of OS-type
whereas the white one corresponds to VS-type.
(c) and (d) are parts of (a) and (b) respectively.
In (a) and (c), the black curves avoid $X$ whereas
in (b) and (d), the white curves avoid $X^c$.  }
\label{fig:zoomin}
\end{figure}

\section{Homogenized Conductivity}

Before we mention the homogenization of the conductivity,
we note that our mixed boundary conditions are connected to  
a Dirichlet problem.

For a subset $A$ in ${\Box_L}$, we consider the reflection to
$\partial_r{\Box_L}$,
$$
R_{\partial_r\Box_L}A:=
\{ 2L-x +\ii y \ |\ x+\ii y \in A\}
$$
and identify the image $R_{\partial_r\Box_L}(A \cap \partial_\ell\Box_L)$
with itself $(A \cap \partial_\ell\Box_L)$.
The above problem can be extended to the differential equation on 
${\Box_L} \cup R_{\partial_r\Box_L}\Box$ with periodic boundary condition.
Since the Neumann conditions in our systems
are naturally satisfied in the geometrical setting,
our problem can be formulated on an annuls
$A_{L_1, L_2}=\{ \xi \in \CC \ | \ L_1 \le |\xi|\le L_2\} \subset \RR^2$
as Dirichlet boundary problem.
Hence we can apply the results \cite{GiTr,Tr73,Tr77}
 for the Dirichlet boundary condition
to our mixed boundary problems (\ref{eq:0-1}) with (\ref{eq:0-3}).
More precise arguments among the boundary conditions in homogenization
is in Ref.~\cite{MK}.

\bigskip
For $\hXconf \in \hcQ(\Box_L)$, we define
the homothety transformation 
$\hat D_a \hXconf = \{a z_i \ | \ z_i \in \hXconf \}$ for $a>0$,
and $\hat D_a \Xconf = \overline{U_{ar}(\hat D_a \hXconf)}$: 
for every $a>0$,
$$
\cE_{\Box_L}[\gamma^\Xconf]= \cE_{\Box_{aL}}[\gamma^{\hat D_a\Xconf}].
$$
Further as mentioned before, we note that our Boolean model 
$(\cQ_r(\Reg), \cB(\cQ_r(\Reg)), \Prob{\lambda})$ 
 has ergodic properties \cite{MR,SKM}.

\bigskip

In order to consider a fixed configuration $X$ which 
stands a typical configuration in $\cQ_r({\Box_L})$ of $\lambda<\lc$ and
$\lambda>\lc$,
let 
$$
\cQ_{r,\lambda}(\CC):= 
\bigcap_{ U \in \cB(\cQ_r(\CC)), \ \Prob{\lambda}(U) = 1 }U,
$$
and $\cQ_{r,\lambda}({\Box_L}) := 
\{X \cap {\Box_L} \ | \ X \in \cQ_{r,\lambda}(\CC)\}$.
For fixed $\lambda$ and $X\in \cQ_{r, \lambda}(\CC)$, 
we are concerned with $\Gamma_{B_{nL}}(X)$ of $n\to \infty$.

\bigskip

In the homogenization problem, the asymptotic expansion 
of
\begin{equation}
\GammaCS{\varepsilon}(\Xconf) := 
\mathrm{inf} \frac{1}{(\uCS{0})^2}\int_{\Box_L}   \
     \gammaCS{\Xconf}\left(\frac{z}{\varepsilon}\right)\nabla \vG(z)^2 d^2 z
\label{eq:epsilonGamma}
\end{equation}
with respect to $\varepsilon$,
$\GammaCS{\varepsilon}(\Xconf) = \GammaCS{0}(\Xconf) + 
\GammaCS{1}(\Xconf) \varepsilon+ 
\GammaCS{2}(\Xconf) \varepsilon^2+ 
\cdots$, is investigated.
The constant term $\GammaCS{0}(\Xconf)$
is the homogenized or the effective conductivity \cite{Koz},
 which are the same as the 
the asymptotic behavior of the energy functional 
$\cE_{B_{nL},0}[\gammaCS{\Xconf}]$ and homogenized conductivity
$\GammaCS{B_{nL}}(\Xconf)$
for a large $n$ limit of $B_{nL}$.

Homogenization problem is basically for a 
local conductivity $\gammaCS{}$ with a periodic property.
Murant and Tartar proved that in the limit of $\varepsilon \to 0$ in
(\ref{eq:epsilonGamma}),
the solution of the generalized Laplace equation 
weakly converges
which is known as $H$-convergence \cite{MT1,MT2}.
On the other hand,
De Giorgi and Franzoni showed that the energy functional
converges  as the minimization problem which 
is known as $\Gamma$-convergence \cite{EDG}.
Recently Ansini, Dal Maso and Ida Zeppieri proved that both are
equivalent for such Dirchlet problems \cite{ADI}.

Further based upon $\Gamma$-convergence, if the configuration
satisfies the ergodic property, Dal Maso and Modica showed
that the energy functional converges for the large $n$ limit \cite{DM}.  
Bourgeat, Mikeli\'c and Wright proved a stronger theorem
in a two-scale convergence \cite{BMW}.

Since the Poisson point process has ergodic property,
these results guarantee the convergence of the conductivity
$\GammaCS{{\Box_{nL}}}(\Xconf\cap B_{nL})$
and the energy functional $\cE_{B_{nL},0}[\gammaCS{\Xconf\cap B_{nL}}]$
for a configuration $X \in \cQ_{r,\lambda}(\CC)$
by taking the large $n$ limit of $B_{nL}$.
The asymptotic behaviors are defined by
\begin{equation}
\cE_{0}[\gammaCS{\Xconf}]:=\lim_{n\to \infty}
\cE_{B_{nL},0}[\gammaCS{\Xconf\cap B_{nL}}],
\quad
\GammaCS{{}}(\Xconf) := \lim_{n \to \infty}
\GammaCS{{\Box_{nL}}}(\Xconf\cap B_{nL}).
\end{equation}

\bigskip

The ergodic theorem as Theorem 6.2 in Ref.~\cite{SKM} can be
applicable to our problem. 
In our Boolean model 
$(\cQ_r(\Reg), \cB(\cQ_r(\Reg)), \Prob{\lambda})$,
Refs.~\cite{BMW} and \cite{DM}
(Examples in Refs.~\cite{DM})
mean the following proposition:
\begin{proposition} \label{prop:LNT}
For $\Xconf_0 \in \cQ_{r}(\Reg)$,
$$
\GammaCS{}(\Xconf_0)
=\lim_{n\to \infty}\Exp{\lambda}
\left(\GammaCS{\Box_{nL}}(\Xconf\cap B_{nL}) \right),
\quad \Prob{\lambda}\mbox{-a.s.}.
$$
\end{proposition}

The right hand side is also denoted by
$\lim_{n\to \infty}\Exp{\lambda}\left(\GammaCS{B_{nL}}(\Xconf) \right)$ or
$\Exp{\lambda}\left(\GammaCS{}(\Xconf) \right)$ simply,
which is called homogenized total conductivity,
or {\it{ homogenized conductivity}}. Due to the Proposition \ref{prop:LNT},
we also call $\GammaCS{{}}(\Xconf)$ homogenized conductivity
$\Prob{\lambda}$ a.s..

\bigskip
\begin{remark} \label{rmk:Gamma(p)}
{\rm{
We note the homothety in our system and
the energy functional as a function of $\cQ_r(B_L)$.
For a given $X \in \cQ_{r, \lambda}(B_L)$ and $r'>0$,
there exist $\lambda'$, 
$X' \in \cQ_{r', \lambda'}(B_{r'L/r})$ and $r'>0$
such that $X = D_{r'/r} X'$.
Hence
for a positive number $r'$, there is an intensity $\lambda'$
such that
\begin{equation}
\lim_{n\to \infty} 
{\mathbf{E}}_{\lambda}(\GammaCS{\Box_{nL}}
(\Xconf\in \cQ_r(\CC)))=
\lim_{n\to \infty} 
{\mathbf{E}}_{\lambda'}(\GammaCS{\Box_{nr'L/r}}
(\Xconf \in \cQ_{r'}(\CC))).
\label{eq:Gamma_r,r'}
\end{equation}
Since the volume fraction is also given by
$$
\lim_{n\to \infty} 
{\mathbf{E}}_{\lambda}(\Vol_{\Box_{nL}}
(\Xconf\in \cQ_r(\CC)))=
\lim_{n\to \infty} 
{\mathbf{E}}_{\lambda'}(\Vol_{\Box_{nr'L/r}}
(\Xconf \in \cQ_{r'}(\CC))),
$$
$\lambda$ is  a function of $r$ and $p$ for our model,
which is denoted by $\lambda = \lambda(p, r)$.
In other words, 
for a fixing the radius $r$,
${\mathbf{E}}_{\lambda(p,r)}(\GammaCS{}(\Xconf))$ is a function of
the volume fraction $p$, as it is well-known.
}}
\end{remark}

From Propositions \ref{prop:LNTp} and \ref{prop:LNT},
$\Exp{\lambda}(\GammaCS{{}}(\Xconf))$ is a function of $p(\lambda, r)$.
We have the following proposition:

\begin{proposition}\label{prop:monoC}
\begin{enumerate}
\item $\Exp{\lambda(p,r)}(\GammaC{}(\Xconf_{r}))$ is a monotonic increasing
function of the volume fraction $p \in [0,1]$. 

\item $\Exp{\lambda(p,r)}(\GammaS{}(\Xconf_{r}))$ is a monotonic decreasing
function of the volume fraction $p \in [0,1]$. 
\end{enumerate}
\end{proposition}

\begin{proof} Let us consider (1).
For $\hat X \in \hcQ(\CC)$, let
$X_r :=\overline{\cup_{z\in \hat X}U_{z,r}}$,
$X_{r-\varepsilon} :=\overline{\cup_{z\in \hat X}U_{z,r-\varepsilon}}$
and $
\delta \gammaC{\Xconf_{r},\varepsilon}:=\gammaC{\Xconf_{r}}- 
\gammaC{\Xconf_{r-\varepsilon}}$. 
The relation between
$\cE_{B_L}(X_r)$ and $\cE_{B_L}(X_{r -\varepsilon})$
is computed;
\begin{eqnarray*}
\min_{\vG \in \cH({\Box_L}|\mathrm{BC}_{B_{L}}^\CS,\uG_0^\CS)}
 \cE[\gammaC{\Xconf_r}; \vG] 
&=&
\int_{\Box_{L}} \gammaC{\Xconf_r} (\nabla \uC{})^2 d^2 z,\\
&\ge&
\int_{\Box_{L}} \gammaC{\Xconf_{r-\varepsilon}} 
(\nabla \uC{})^2 d^2 z+
\int_{\Box_{L}} \delta\gammaC{\Xconf} 
(\nabla \uC{})^2 d^2 z \\
&\ge&
\min_{\vG \in \cH({\Box_L}|\mathrm{BC}_{B_{L}}^\CS,\uG_0^\CS)}
 \cE[\gammaC{\Xconf_{r-\varepsilon}}; \vG] .
\end{eqnarray*}
Since 
$\int_{\Box_{L}} \delta\gammaC{\Xconf} (\nabla \uC{})^2d^2z$ 
must be positive, 
$\GammaC{\Xconf_{r}}\ge \GammaC{\Xconf_{r-\varepsilon}}$.
Since   
$\Prob{\lambda}$-almost surely
$\GammaCS{{}}(\Xconf)$ and $p(\Xconf)$ 
equal to their expectation values from 
Propositions \ref{prop:LNTp} and \ref{prop:LNT},
we have the result.
(2) is also obtained in the similar
computation to the above estimation.
\end{proof}

For the volume fraction $p\in [0,1]$,
we define the {\it{conductivity curves}},
\begin{equation}
\GammaC{}(p):=\Exp{\lambda(p,r)}(\GammaC{}(\Xconf)),
\quad
\GammaS{}(p):=\Exp{\lambda(p,r)}(\GammaS{}(\Xconf)).
\end{equation}


\bigskip

From Proposition\ref{prop:Roy},
the following corollary is obtained:

\begin{corollary}\label{cor:Penrose}
\begin{equation}
\lim_{\cond_\inf \to +0}\GammaC{}(p)= 
\displaystyle{
\left\{\begin{array}{ll}
=0& \mbox{ for } p < \pc,\\
\neq 0,  & \mbox{ for } p > \pc,\\
\end{array}\right.
}
\nonumber
\end{equation}
\begin{equation}
\lim_{\cond_\inf \to +0}\GammaS{}(p)= 
\displaystyle{
\left\{\begin{array}{ll}
\neq 0& \mbox{ for } p < \pc,\\
= 0,  & \mbox{ for } p > \pc.\\
\end{array}\right.
}
\nonumber
\end{equation}
\end{corollary}


\section{Quasiconformal Properties of Binary System}

\subsection{Quasiconformal Properties and 
Keller-Dykhne Reciprocity Law}

Let us fix a configuration $X\in \cQ_r({\Box_L})$, and
consider
the solutions $\uCS{}$ of (\ref{eq:0-1}) with (\ref{eq:0-3})
from quasi-harmonic map theory \cite{AIM}.
In this section, we consider the complex-valued generalized function,
$$
\psi^X := \uS{} + \ii \uC{} \in \cD({\Box_L}, \CC),
$$
 and  the Beltrami coefficient,
$$
\displaystyle{
\mu^{{\Box_L}}_{\psi^X}(z):=
\frac{\overline{\partial}\psi^X(z)}
{\overline{\partial\psi^X(z)}},
}
 \in \cD({\Box_L}, \CC),
$$ 
where $\displaystyle{
\partial = \frac{1}{2}(\partial_x - \ii \partial_y)}$,
$\displaystyle{
\overline\partial = \frac{1}{2}(\partial_x + \ii \partial_y)}$,
$\displaystyle{
\partial_x = \frac{\partial}{\partial x}}$, and 
$\displaystyle{
\partial_y = \frac{\partial}{\partial y}}$.

The theory of the quasiconformal mappings \cite{AIM} provides
the following proposition:

\begin{proposition}\label{prop:qconf}
For a measurable non-negative valued function $\gamma$ on ${\Box_L}$,
a weak solution $v$ 
in $\cH({\Box_L}|\mathrm{BC}^\OS_{B_L},\uG_0^\OS)$ of
 the elliptic differential equation,
$$
               \nabla \gamma \nabla v = 0,
$$
has $\gamma$-harmonic conjugate $u$
in $\cH({\Box_L}|\mathrm{BC}^\VS_{B_L},\uG_0^\VS)$, i.e.,
$$
\partial_x u =  \gamma \partial_y v,
\quad
\partial_y u = -\gamma \partial_x v,
$$
which $u$ satisfies
$$
               \nabla \frac{1}{\gamma} \nabla u = 0.
$$
The function $f := u + \ii v$ and
the Beltrami coefficient $\mu := \dfrac{1-\gamma}{1+\gamma}$
obey
$$
      \overline{\partial} f = \mu \overline{\partial f}.
$$
\end{proposition}

Noting the relation,
$$
\gammaS{X}(z) =\frac{\cond_\inf \cond_\mat}{\gammaC{X}(z)},
$$
we apply Proposition \ref{prop:qconf} to our system.
In other words,
we regard our system associated with the $\gamma$-harmonic conjugation
by letting
\begin{equation}
\gamma:=\alpha \gammaC{X}(z),
\quad
u:=\beta\uS{},\quad
v:=\uC{}, \quad
\label{eq:alpha's}
\end{equation}
where $\alpha$ and $\beta$ are constant positive numbers,
i.e., the correspondence between systems of OS-type and VS-type
can be regarded as the $\gamma$-harmonic conjugation.

It is obvious that above
boundary conditions are consistent with
 $\gamma$-harmonic relation because they are orthogonal,
$$
(\nabla \uC{}, \nabla \uS{}) = 0
\label{eq:Vp,Vp=0}
$$
everywhere.
For $\uC{0}=\uS{0}$ case, we illustrate examples
in Figure \ref{fig:zoomin}
which are the numerical computational results as mentioned in Section 5
and  show that  the equipotential curves of OS-types and VS-types
cross perpendicularly.
We have the following proposition:

\begin{proposition}
$\psi^X=\uS{}+\ii\uC{}$ is a quasiconformal map from
$B_L$ to 
$I_{\uS{0}}^{\uC{0}}:=[0,\uS{0}]\times[0,\uC{0}]$ for $\gamma_0>0$.
\end{proposition}

Let us define the natural maps for $u \in [0,\uS{0}]$
and $v \in [0,\uC{0}]$
$$
\iota_u : [0,\uC{0}] \to I_{\uS{0}}^{\uC{0}},
\quad(v\mapsto (u,v)), \quad
\iota_v : [0,\uS{0}] \to  I_{\uS{0}}^{\uC{0}},
\quad(u\mapsto (u,v)). \quad
$$

\begin{lemma}\label{lem:eqicurve}
For each $u \in [0,\uS{0}]$ and $v \in [0,\uC{0}]$,
the equipotential curves $\CVS{u}(X)$ of the potential $\uS{}=u$
and $\COS{v}(X)$ of the potential $\uC{}=v$
are given by
$$
\CVS{u}(X)={\psi^X}^{-1} \iota_u[0,\uC{0}]
\quad\mbox{and}\quad
\COS{v}(X)={\psi^X}^{-1} \iota_v[0,\uS{0}].
$$
\end{lemma}

\begin{remark}{\rm{
\begin{enumerate}
\item
We should note that the families of the equipotential curves,
$\{\CVS{u}(X)={\psi^X}^{-1} \iota_u[0,\uC{0}] \ | \ u = u_0, u_1, \ldots\}$ 
and
$\{\COS{v}(X)={\psi^X}^{-1} \iota_v[0,\uS{0}] \ | \ v = v_0, v_1, \ldots\}$ 
are given as non-intersection curves in ${\Box_L}$. 

\item
As mentioned in Sec.5, we numerically computed the equipotential
curves $\CVS{u}(X)$ and $\COS{v}(X)$ in 
 Figure \ref{fig:zoomin} and Figure \ref{fig:contours}.
In Figure \ref{fig:contours},
by numerically solving (\ref{eq:0-1})
 we obtain the
potential distributions $\uC{}(x,y)$ and $\uS{}(x,y)$
for each volume fraction $p =  0.2, 0.6, 0.9$,
$\uC{0}=\uS{0}=1$ and a seed $i_s$ of the pseudo-randomness. 
We display the equipotential curves whose interval 
$\delta \uG$ equals $0.1$  there.
The curves are very complicated for the $p=0.6$ case
which is near the threshold $p_c$.
With Figure \ref{fig:zoomin}, these curves of
OS-type and VS-type show their duality.
\end{enumerate}
}}
\end{remark}

\begin{figure}[htbp]
\begin{center}
\includegraphics[width=12cm]{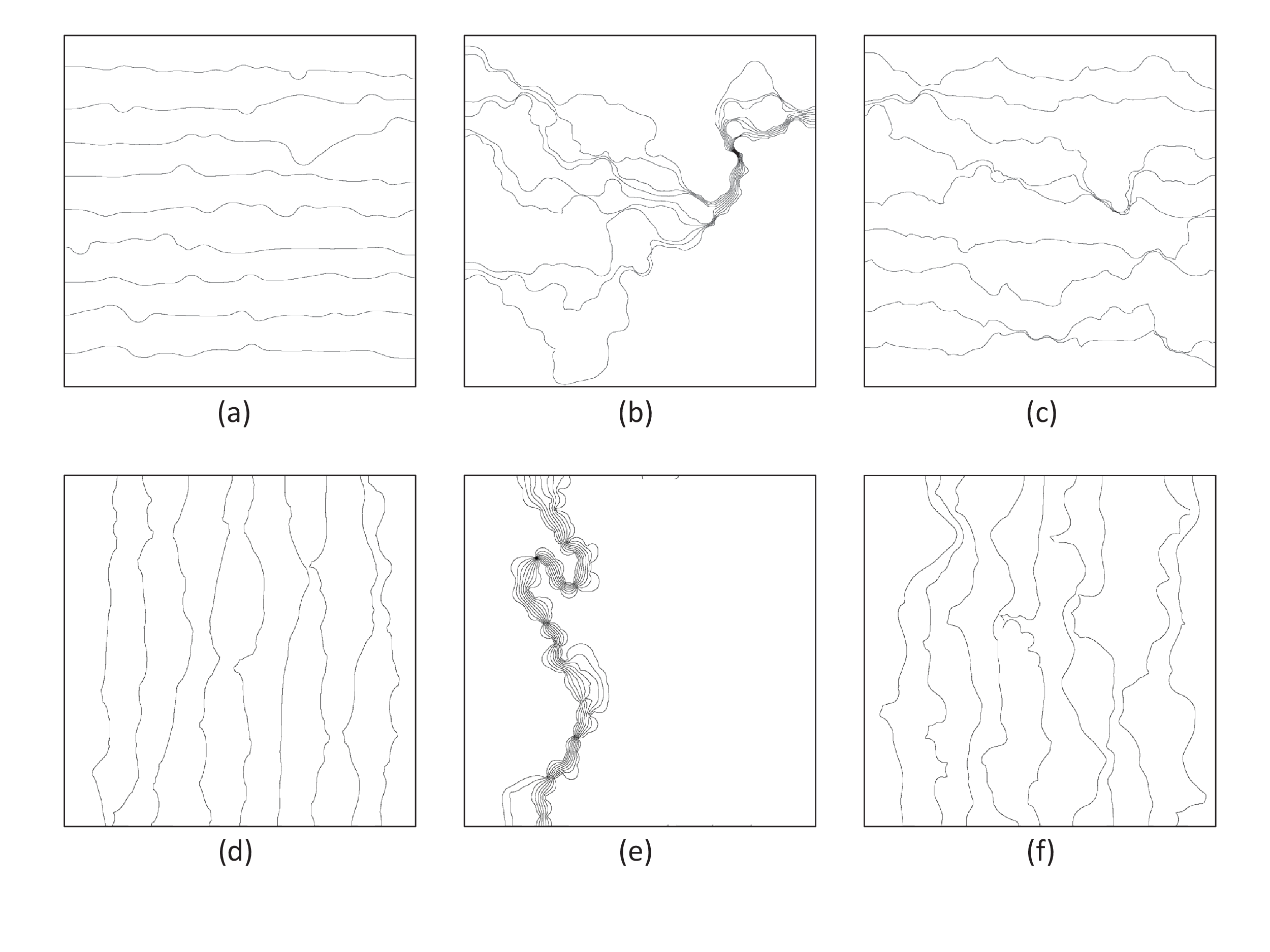}
\end{center}
\caption{Equipotential curves of the potential distributions
of the seed $i_s = 1$ with
with the volume fractions 0.2, 0.6, 0.8:
(a), (b), (c) are for OS-type and (d), (e), (f) for VS-type respectively.
The curves correspond to the values $\uG = $ 0.1, 0.2,
$\ldots$, 0.8, and 0.9.  The interval of the curves $\delta \uG$ is $0.1$.
}
\label{fig:contours}
\end{figure}

\bigskip

Due to the duality of $\gamma$-harmonic property,
we have the following well-known results of the
Keller-Dykhne reciprocity law of the conductivity \cite{Dyk,Kell}:
\begin{proposition} \label{rmk:gamma-conjugate}
\begin{enumerate}
\item
For every $X\in \cQ_r(\Box_L)$, 
$
\GammaS{{\Box_L}}(X)
\GammaC{{\Box_L}}(X)
= \cond_\inf \cond_\mat
$,

\item
$\Exp{\lambda}(\GammaC{}(\Xconf))
\Exp{\lambda}(\GammaS{}(\Xconf)) = \cond_\mat\cond_\inf$ 
for every $\lambda$, and

\item 
$\GammaC{}(p) \GammaS{}(p) = \cond_\mat\cond_\inf$
for every $p \in [0,1]$. 
\end{enumerate}
\end{proposition}

\begin{proof}
Since the direct computations show
$$
\uC{0}\GammaC{{\Box_L}}(X)=
\int_{\partial_\up{\Box_L}}
 \gammaC{X} \frac{\partial}{\partial y} \uC{}(z)dx
=
\frac{\beta}{\alpha}
\int_{\partial_\up{\Box_L}} 
 \frac{\partial}{\partial x} \uS{}(z)dx =
\frac{\beta}{\alpha} \uS{0},
$$
$$
\uS{0}\GammaS{{\Box_L}}(X)=
\int_{\partial_\rght{\Box_L}}
 \gammaS{X} \frac{\partial}{\partial x} \uS{}(z)dy
= 
\gamma_0 \gamma_1
\frac{\alpha}{\beta}
\int_{\partial_\rght{\Box_L}} 
 \frac{\partial}{\partial x} \uC{}(z)dy
=\gamma_0 \gamma_1
\frac{\alpha}{\beta}
\uC{0},
$$
(1) is obvious.  Due to
Proposition \ref{prop:LNT}, we have (2) and (3).
\end{proof}

As Kozlov showed for the random checkerboard model 
in Theorem 8 of Ref.~\cite{Koz},
it is expected that
the effective conductivity of the high contrast
local conductivity for a more general model also behaves
like $\sqrt{\gamma_0\gamma_1}$ at the threshold.
Thus we show that $\GammaCS{}(p)$ is identical to 
$\sqrt{\gamma_0\gamma_1}$ at $p_c$ by 
assuming that $\GammaCS{}(p)$ is a continuous function over $p\in [0,1]$,
which is also numerically shown in Remark \ref{rmk:R_cond2}.

We should have the following lemma.
\begin{lemma}\label{lmm:sqrt}
Assume the continuous property of  $\GammaCS{}(p)$.
Then there exists
a monotonic increasing continuous function $h$ over 
$[0,1]$
such that $h(0) = \sqrt{\gamma_0/\gamma_1}$,
 $h(p_c) = 1$,
$h(1) = \sqrt{\gamma_1/\gamma_0}$,
$$
\GammaC{}(p) = \sqrt{\cond_\mat\cond_\inf} h(p)
\quad\mbox{ and }\quad
\GammaS{}(p) = \sqrt{\cond_\mat\cond_\inf} /h(p).
$$
\end{lemma}

\begin{proof}
The region of $\GammaCS{}(p)$ is $[\gamma_\inf,\gamma_\mat]$.
It is obvious that $\GammaC{}(0)=\GammaS{}(1) = \gamma_\inf$
and  $\GammaC{}(1)=\GammaS{}(0) = \gamma_\mat$.
They are monotonic continuous functions from the assumption.
They must cross at a point $p_0$, i.e.,
$\GammaC{}(p_0)=\GammaS{}(p_0)$.
Corollary \ref{cor:Penrose} means that $p_0$ must be $p_c$.
From the Keller-Dykhne reciprocity law in 
Proposition \ref{rmk:gamma-conjugate}
we have $\GammaC{}(p_c)=\GammaS{}(p_c) = \sqrt{\gamma_\inf\gamma_\mat}$.
The assumption asserts the existence of the function.
\end{proof}

\bigskip

For the case of $\beta = 1$ in (\ref{eq:alpha's})
we have the following conformal relation;
\begin{proposition}\label{them:complex}
For $\Xconf \in \cQ_r({\Box_L})$,
we have the conformal properties,
\begin{equation}
\supp(\mu^{\Box_L}_{{\psi^X}}) 
= \displaystyle{
\left\{\begin{array}{ll}
\Xconf & \mbox{ for } 
\alpha = \dfrac{1}{\cond_\inf}, 
\uS{0}= \dfrac{\GammaC{{\Box_L}}(X)}{\cond_\inf}\uC{0}, \\
\overline{\Xconf^c}   & \mbox{ for } 
\alpha = \dfrac{1}{\cond_\inf}, 
\uC{0}= \dfrac{\GammaS{{\Box_L}}(X)}{\cond_\mat}\uS{0}.
\end{array}\right.
}
\label{eq:CRrel} 
\end{equation}
\end{proposition}

\begin{proof}
$\beta=\dfrac{\uC{0}}{\uS{0}}
\dfrac{\alpha}{\GammaC{{\Box_L}}(X)}=\dfrac{\uC{0}}{\uS{0}}
\dfrac{\alpha\gamma_0\gamma_1}{\GammaS{{\Box_L}}(X)}=1$.
\end{proof}

\begin{figure}[htbp]
 \begin{center}
  \includegraphics[height=6cm]{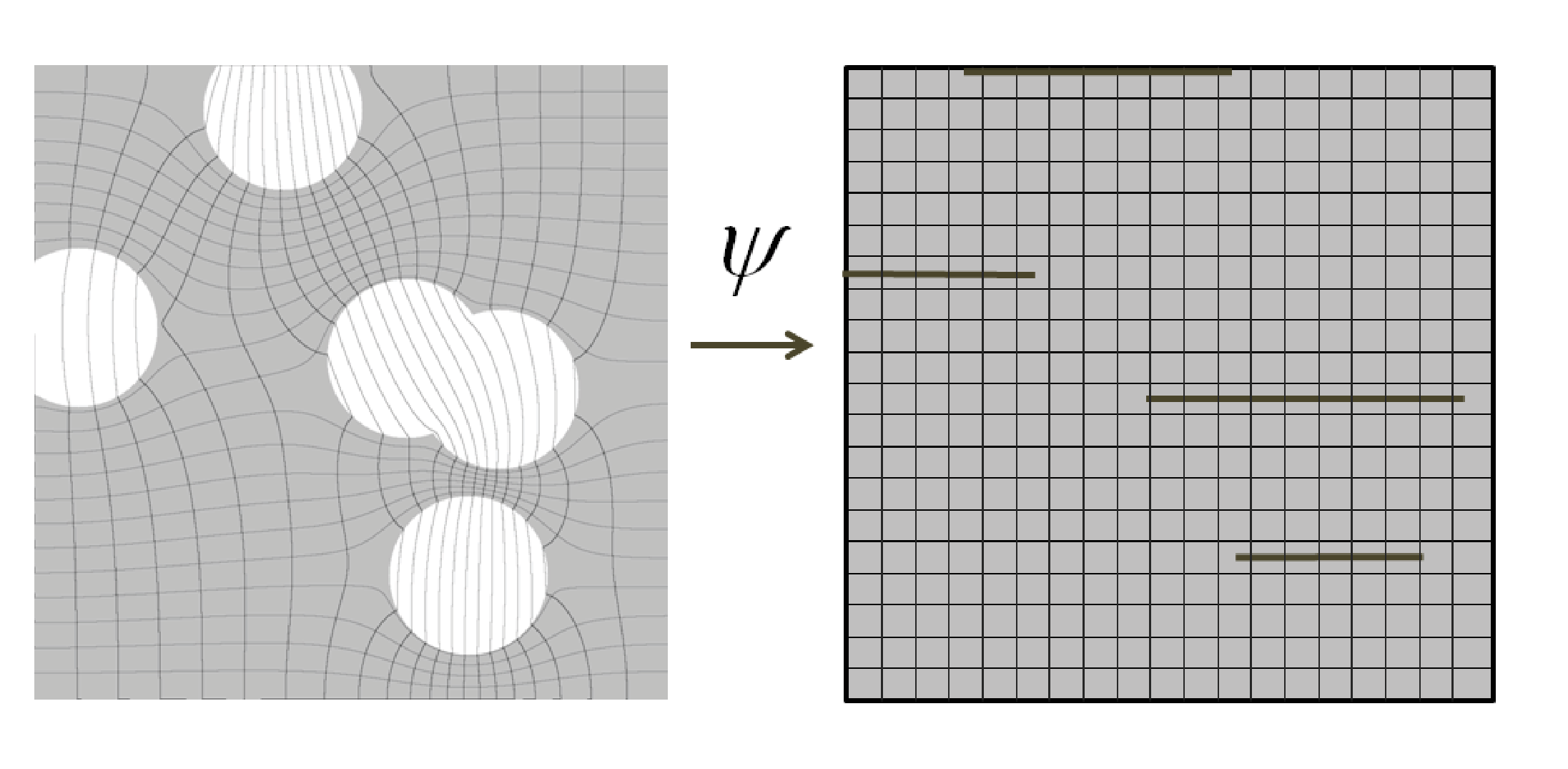} \newline
 \end{center}
\caption{
The quasiconformal map
}
\label{fig:CMap}
\end{figure}

\subsection{Quasiconformal and Conformal Maps for Infinitesimal $\gamma_\inf$}

Let us introduce some geometrical terms \cite{MR}.
For $\Xconf\in \cQ_r({\Box_L})$ and two disjoint regions $A_1$ and
$A_2 \in {\Box_L}$, we say that a continuous curve $C \subset {\Box_L}$ 
is an occupied
(resp. a vacant) connection of $A_1$ and $A_2$ 
with respect to $\Xconf$, if $C \cap A_1 \neq \emptyset$,
$C \cap A_2 \neq \emptyset$,  and $C \subset \Xconf$ 
(resp. $C \subset \overline{\Xconf^c}$). 

Further we say that
such $A_1$ and $A_2$ are occupied 
(resp.  vacant) connected
with respect to $\Xconf$, if there is 
a continuous curve $C$ which is an occupied
(resp. vacant) connection of $A_1$ and $A_2$
with respect to $\Xconf$.
We also say that
such $A_1$ and $A_2$ are occupied 
(resp.  vacant) disconnected
with respect to $\Xconf$, if there is no 
 continuous curve $C$ which is an occupied
(resp. vacant) connection of $A_1$ and $A_2$
 with respect to $\Xconf$.

\begin{definition}
\begin{enumerate}
\item
For $\Xconf\in \cQ_r({\Box_L})$ 
we say that 
$\Xconf^\circ$ (resp. $\Xconf^c$) is {\it{totally 
connected with}} $\partial \Box_L$ if 
every $\partial_\alpha {\Box_L}$ ($\alpha=u,d,r,\ell$) are mutually
occupied (resp. vacant) connected.

\item
For $\Xconf\in \cQ_r({\Box_L})$ 
we say that 
$\Xconf^\circ$ (resp. $\Xconf^c$)  is {\it{totally disconnected with}} 
$\partial \Box_L$ if 
every $\partial_\alpha {\Box_L}$ ($\alpha=\up,\dwn,\rght,\lft$) are mutually
occupied (resp. vacant) disconnected.

\item
We say that $\Xconf$ {\it{dominates}} ${\Box_L}$ if 
$\Xconf^\circ$ is 
totally connected with $\partial \Box_L$
and $\Xconf^c$ is 
totally disconnected with $\partial \Box_L$.

\item
We say that $\Xconf^c$ {\it{dominates}} ${\Box_L}$
 if $\Xconf^c$ is 
totally connected with $\partial \Box_L$
and $\Xconf$ is 
totally disconnected with $\partial \Box_L$.
\end{enumerate}
\end{definition}

Let us define the subset of $\cQ_{r}({\Box_L})$ which
$$
  \cQ_{<}({\Box_L}):=
      \{ X \in \cQ_{r}({\Box_L}) \ | \ X^c\mbox{ dominates } {\Box_L}\}
          \subset \cQ_r({\Box_L}),
$$
$$
  \cQ_{>}({\Box_L}):=\{ X \in \cQ_{r}({\Box_L}) \ | \ X\mbox{ dominates } {\Box_L}\}
          \subset \cQ_r({\Box_L}). 
$$

As in Figure \ref{fig:zoomin}, we show that the domain of the
equipotential curves.

\begin{proposition} \label{prop:const}
\begin{enumerate}
\item  For $X \in \cQ_<({\Box_L})$, 
let the decomposition of  $\Xconf$ by the connected parts
$\Xconf = \coprod_{i} {\Xconf}^{(i)}$.
In the limit of $\cond_\inf \to 0$,
the solution of (\ref{eq:0-1})  with $\mathrm{BC}^+_{B_L}$,
$\uC{}$ is constant in each ${\Xconfi{i}{}}$
and
every equipotential curve $\COS{v}$ of $v\in [0,\uC{0}]$
 does not penetrate
into each ${{\Xconf}^{(i)}}^\circ$.

\item  For $X \in \cQ_>({\Box_L})$, 
let the decomposition of  $\Xconf^c$ by the connected parts
$\Xconf^c = \coprod_{j} {\Xconf^c}^{(j)}$.
In the limit of $\cond_\inf \to 0$,
the solution of (\ref{eq:0-1})  with $\mathrm{BC}^-_{B_L}$
$\uS{}$ is constant in each ${\Xconf^c}^{(j)}$
every equipotential curve $\CVS{u}$ of $u\in [0,\uS{0}]$
 does not penetrate
into each ${\Xconf^c}^{(j)}$.
\end{enumerate}
\end{proposition}

\begin{proof}
Let $X \in \cQ_<({\Box_L})$. Then it is not connected with
$\partial_a \Box$ and thus 
$\lim_{\gamma_0\to 0}\GammaC{{\Box_L}}(X)=0$.
Hence there is a positive number $a_0(>0)$ given by
$$
	a_0:= \sup \{a\in (0,1) \ | \  
\lim_{\gamma_0\to 0}\GammaC{{\Box_L}}/\gamma_0^{a}=0\}.
$$
Let $\cE_{B_L,0<}^{\OS}(X):=\GammaC{{\Box_L}}(X)\uG_0/\gamma_\inf^{a_0}$.
$\cE_{B_L,0<}^{\OS}(X)$ is decomposed like
\begin{equation}
\cE_{B_L,0<}^{\OS}(X)
=\frac{\cond_\mat}{\cond_\inf^{a_0}}\sum_{i}\int_{{\Xconfi{i}{}}}   
(\nabla \uG^\OS)^2 d^2 z
+ \gamma_\inf^{1-a_0}\int_{{\Xconf_{}^c}}  (\nabla \uG^\OS)^2 d^2 z.
\end{equation}
The energy functional
$\cE_{B_L,0<}^{\OS}$ must vanish for $\gamma_0 \to 0$
but $\GammaC{B_L}(X)$ has the non-trivial solution;
thus we have assumed $a_0 < 1$.
Since for the limit of the $\cond_\inf \to +0$,
the first term must vanish
and $(\nabla \uG)^2$ vanish on ${\Xconfi{i}{}}^\circ$.
Hence $\uC{}$ in each $\Xconfi{i}{}$ becomes a constant function
$\uC{i}$ for $\gamma_0\to0$.
(2) is the same as (1).
\end{proof}

\bigskip 

Let us consider the lateral slit model $\cS_{<}(I^{\uC{0}}_{\uS{0}})$ 
which is given
as a point process $\hcQ(I^{\uC{0}}_{\uS{0}})$ 
with lateral slits.
For $\hXconf \in \hcQ(I^{\uC{0}}_{\uS{0}})$
($\hXconf =\{(x_i,y_i)\ | \ i = 1,2,\ldots,m\}$)
we consider random variables $c=(c_i)_{i = 1,\ldots,m} \in [0,\uS{0}]^m$
and a set of slits
$$
    S_{\hXconf, c}:=\{s_{(x_i-c_i, y_i),(x_i+c_i, y_i)} \cap 
I^{\uC{0}}_{\uS{0}} 
 \ |\ (x_i, y_i) \in \hXconf\},
$$
where  $s_{z,z'}$ is the segment between $z$ and $z'$.
The element $\fs$ of $\cS_{<}(I^{\uC{0}}_{\uS{0}})$ is given by
$\fs:=I^{\uC{0}}_{\uS{0}}\setminus S_{\hXconf, c}$.

As Figure \ref{fig:CMap} displays the map ${\psi^X}$
for $X\in \cQ_{<}({\Box_L})$,
there is an element $\fs=I^{\uC{0}}_{\uS{0}}\setminus S_{\hXconf, c}$
of $\cS_{<}(I^{\uC{0}}_{\uS{0}})$ and 
we have the quasiconformal map $\psi^X|_{X^c}$ from $X^c \to \fs$.
In other words,
$\psi^X|_{\overline{X^c}}$ is a map 
from $\overline{X^c} \to I^{\uC{0}}_{\uS{0}}\equiv \overline{\fs}$
and each slit $s_i \in S_{\hXconf, c}$ corresponds to a connected
element $X^{(i)}$ in $\Xconf$, i.e.,
$\partial X^{(i)} = {\psi^X}^{-1}(s_i)$;
$X^c$ is quasiconformally equivalent to $\fs$.

The length of each slit and the configuration
depend on the configuration of the disk in each 
$\Xconfi{i}{}$ and others $\Xconfi{j}{}$ $(j\neq i)$.
It implies that 
the configuration $X\in \cQ_{<}({\Box_L})$ is completely
parameterized by the element of $\cS_{<}(I^{\uC{0}}_{\uS{0}})$.
In other words, there is a random slit model 
on a certain probability space
$(\cS_{<}(I^{\uC{0}}_{\uS{0}}), 
\cB(\hcQ(I^{\uC{0}}_{\uS{0}}))
\times\cB([0,\uS{0}]), \tilde{\mathbf{P}}_{r,\lambda})$
which corresponds to our Boolean model for $\lambda$ and $r$.

Similarly the vertical slits model are defined as
$\cS_{>}(I^{\uC{0}}_{\uS{0}})$ whose element is 
given as $X \in \cQ(I^{\uC{0}}_{\uS{0}})$ with vertical slits
 with random length.

Then we have the following proposition.

\begin{proposition}\label{prop:mapQS}
In the limit $\cond_\inf \to 0$,
the function ${\psi^X}$ induces injections 
$\Psi_<:\cQ_{<}({\Box_L})\to \cS_{<}(I^{\uC{0}}_{\uS{0}})$
and $\Psi_>:\cQ_{>}({\Box_L})\to \cS_{>}(I^{\uC{0}}_{\uS{0}})$.
\end{proposition}

\begin{proof}
Proposition \ref{prop:const} show the results.
\end{proof}

By considering the $n$ large limit of $B_{nL}$, 
as the behavior of $\cQ(\CC)$,
we have the main theorem:
\begin{theorem}\label{mapQS}
\begin{enumerate}
\item
For $\lambda<\lambda_c$
there is an element $\fs \in \cS_{<}(I^{\uC{0}}_{\uS{0}})$
such that by $\psi^X$, 
$\displaystyle{X^c\bigcap\left(
\bigcup_{n=1}^\infty B_{nL}\right)}
$ is quasiconformally equivalent to
$\fs$,
$\Prob{\lambda}$-a.s..

\item
For $\lambda>\lambda_c$, 
there is an element $\fs \in \cS_{>}(I^{\uC{0}}_{\uS{0}})$
such that by $\psi^X$,  
$\displaystyle{X^\circ\bigcap\left(
\bigcup_{n=1}^\infty B_{nL}\right)}
$ is quasiconformally equivalent to
$\fs$, $\Prob{\lambda}$-a.s..
\end{enumerate}
\end{theorem}

\begin{proof}
We show that $\psi^X(X^c)$ belongs to $\cS_{<}(I^{\uC{0}}_{\uS{0}})$.
Let us consider
$$
  \cQ_{<,B_L}(\Reg):=
      \{ X \in \cQ_{r}(\Reg) \ | \ X^c\cap \Box_L \mbox{ dominates } {\Box_L}\}
          \subset \cQ_r(\Reg).
$$
Then $\cQ_{<,B_{nL}}(\Reg) \subset \cQ_{<,B_{n'L}}(\Reg)$ 
for $n< n'$. 
Due to Corollary \ref{cor:Penrose} and 
Proposition \ref{prop:mapQS}, 
for $\lambda<\lambda_c$
we have the relation,
$$
\cQ_{r, \lambda}(\Reg) \subset \bigcup_{n=1}^\infty\cQ_{<,B_{n L}}(\Reg), 
$$
and the relation (1). The second one is also obtained similarly.
\end{proof}

\subsection{Conformal Map and Fractal Structure}

In this subsection, let us consider the conformal map by letting 
$\beta =1$ of (\ref{eq:alpha's}) for $\gamma_0 \to 0$.
For an element $X \in \cQ_r(B_L)$, we consider
$$
\uS{0}=f_0^+(X):= \dfrac{\GammaC{B_L}(X)}{\cond_\inf}\uC{0},
\quad\mbox{or}\quad
\uC{0}=f_0^-(X):= \dfrac{\GammaS{B_L}(X)}{\cond_\inf}\uS{0}.
$$

Then we directly have the following lemma:
\begin{lemma}\label{lmm:confmap}
\begin{enumerate}
\item
For $X \in \cQ_<(B_L)$ and a fixed $\uC{0}(>0)$, the map 
$$
{\psi^X}:\overline{X^c} \to I^{\uC{0}}_{f_0^+(X)}
$$
is conformal and  its image is $\overline{\fs}=I^{\uC{0}}_{f_0^+(X)}$
of a certain $\fs \in \cS_<(I^{\uC{0}}_{f_0^+(X)})$.

The arclength of $\COS{v}(X)$ for each $v \in (0,\uC{0})$ is given by
$$
\int_{[0,f_0^+(X)]} \left|\frac{\partial z}{\partial u}\right|(u,v) d u.
$$

\item
For $X \in \cQ_>(B_L)$ and a fixed $\uS{0}(>0)$, the map 
$$
{\psi^X}: X \to I_{\uS{0}}^{f_0^-(X)}
$$
is conformal and  its image is $\overline{\fs}=I_{\uS{0}}^{f_0^-(X)}$
of a certain $\fs \in \cS_>(I_{\uS{0}}^{f_0^-(X)})$.

The arclength of $\CVS{u}$ for each $u \in (0,\uS{0})$ is given by
$$
\int_{[0,f_0^-(X)]} \left|\frac{\partial z}{\partial v}\right|(u,v) d v.
$$
\end{enumerate}
\end{lemma}

\bigskip

We consider
the large $n$ limit of the map
$
 {\psi^X}: B_{nL} \to I^{\uC{0}}_{\uS{0}}.
$
In the limit, the homogenized conductivity is a 
function of the volume fraction.

We consider the conformal property of $X \in \cQ_r(\CC)$
under the limit by letting
$f_0^+(p):=\GammaC{}(p)\uC{0}/\gamma_0$
and $f_0^-(p):=\GammaS{}(p)\uS{0}/\gamma_0$
$\Prob{\lambda(p)}$-a.s..

Noting Lemma \ref{lmm:sqrt}, we have the following lemma:
\begin{lemma}\label{lmm:eqicurve}
Assume the continuous property of  $\GammaCS{}(p)$.
Then the following properties hold:

\begin{enumerate}
\item $I^{\uC{0}}_{f_{0}^+(p)}$ approaches to 
$[0, \infty]\times[0,\uC{0}]$ for
$p \nearrow p_c$
and $\cond_\inf \to 0$, 
and

\item $I_{\uS{0}}^{f_{0}^-(p)}$ approaches to 
$[0,\uS{0}]\times[0,\infty]$ for
$p \searrow p_c$ and $\cond_\inf \to 0$.
\end{enumerate}
\end{lemma}

\begin{proof}
It is obvious.
\end{proof}

In order to see the meaning of Lemma \ref{lmm:eqicurve}, we will consider
paths $\xi$ in $\cQ_r(\Reg)$:
\begin{eqnarray*}
\Path\cQ_r(\Reg) :=\{
\xi:[0,1] \to \cQ_r(\Reg) & | &\Vol(\xi(p))=p,\  
\xi(p)\in \cQ_{r,\lambda(p)}(\Reg)
\nonumber \\
& & \xi(p) \subset \xi(p') \mbox{ for every }p\le p'\}.
\end{eqnarray*}
For a given  $\xi \in \Path\cQ_r(\Reg)$,
there is an integer $N(\xi)$ such that
for every $n>N(\xi)$, $\GammaC{B_{nL}}(\xi(p)\cap B_{nL})$ is a 
monotonic increasing function of $p \in [0,1]$. 
We fix $n > N(\xi)$.
$$
p_c(\xi,n):= \Inf\{p \in [0,1]\ |\ \lim_{\gamma_0 \to 0} 
\GammaC{B_{nL}}(\xi(p)\cap B_{nL}) \neq 0\}.
$$
$$
p_c^*(\xi,n):= \sup\{p \in [0,1]\ |\ \lim_{\gamma_0 \to 0} 
\GammaS{B_{nL}}(\xi(p)\cap B_{nL}) \neq 0\}.
$$
We also have the conformal map 
$\psi^{\xi(p)}:\xi(p)^c \cap B_{nL} \to 
I^{\uC{0}}_{f_{0}^+(\xi(p))}$ denoted by $\psi_{\xi(p),n}$
for $p \in [0, p_c(\xi,n))$,
and 
the conformal map 
$\psi^{\xi(p)}:\xi(p) \cap B_{nL} \to 
I_{\uS{0}}^{f_{0}^-(\xi(p))}$ denoted by $\psi_{\xi(p),n}^*$
for $p \in (p_c(\xi,n)^*,1]$.

Recalling 
Lemma \ref{lmm:sqrt}, the finite version of
Lemma \ref{lmm:eqicurve} is given by the following theorem:
\begin{theorem} \label{thm:eqicurve}
For $\xi \in \Path\cQ_r(\Reg)$,
$n > N(\xi)$, and $\psi_{\xi(p),n}$,
assume that there are a positive number $\varepsilon$ and
a monotonic increasing continuous functions $g$ and $g^*$ over 
$(-\varepsilon, \varepsilon)$
such that $g(0) = g^*(0)=0$,
$\GammaC{B_{nL}}(\xi(p)\cap B_{nL})
=\sqrt{\cond_\mat\cond_\inf} (1 + g(p-p_c))$ for
$p\in (p_c(\xi,n)-\varepsilon,p_c(\xi,n)+ \varepsilon)$,
 and  
$\GammaS{B_{nL}}(\xi(p)\cap B_{nL})
=\sqrt{\cond_\mat\cond_\inf} (1 - g^*(p-p_c))$ for
$p\in (p_c^*(\xi,n)-\varepsilon,p_c^*(\xi,n)+ \varepsilon)$.
Then

\begin{enumerate}
\item the length $f^+_0(\xi(p))$ of the image of $\COS{v}(X)$ 
 by $\psi_{\xi(p),n}$
diverges for
$p \nearrow p_c(\xi,n)$ and $\cond_\inf \to 0$, and

\item the length $f^-_0(\xi(p))$ of the image of $\CVS{u}(X)$ 
by $\psi_{\xi(p),n}$
diverges for
$p \searrow p_c^*(\xi,n)$ and $\cond_\inf \to 0$.
\end{enumerate}
\end{theorem}

\begin{remark}\label{rmk:obs}{\rm{
Let us give an observation to mention the meaning of 
Theorem \ref{thm:eqicurve} by handling the case of
$p < p_c(\xi,n)$ and $(\alpha, \beta) =(1/\cond_\inf, 1)$
of (\ref{eq:alpha's}) 
under the assumptions.
There we have the two important properties:
\begin{enumerate}

\item The expectation values for $n\to \infty$:
$$
\Exp{\lambda(p,\lambda)}
\left(\left|\frac{\partial \uC{}}{\partial y}\right|(u,v) \right)
 =\frac{u_0^+}{nL},
\quad
\Exp{\lambda(p,\lambda)}
\left(\left|\frac{\partial \uS{}}{\partial x}\right|(u,v) \right)
 =\frac{u_0^+\GammaC{}(p)}{nL \gamma_\inf}.
$$
The first one is finite but the second one diverges for the limit
$p \nearrow p_c$.

\item
The conformal property for $z \in \xi(p)^c$:
$$
    \frac{\partial \uS{}}{\partial z} = \ii \frac{\partial \uC{}}{\partial z} 
, \quad
\left|\frac{\partial \uS{}}{\partial z}\right| = 
 \left|\frac{\partial \uC{}}{\partial z}\right|. 
$$
\end{enumerate}

For $\xi(p)$ of $p\to p_c(\xi,n)$,
both properties mean that there exists, at least, a narrow slit
such that two clusters connected with $\partial_\up \Box_{nL}$ and
$\partial_\dwn \Box_{nL}$ respectively
face each other by distance $d$.
(Figure 2 (b) shows that there is such a slit.)
Then at the slit,
the intensity  $\displaystyle{\left|\frac{\partial \uC{}}{\partial z}\right|}$
has the order of $\dfrac{u_0^+}{d}$.
If $d$ is proportional to $(p_c(\xi,n)-p)^b$ of $b>0$,
the intensity diverges
 for the limit $p \nearrow p_c(\xi,n)$ and $\gamma_\inf \to 0$.

Let us assume that there are such slits in $\Box_{nL}$. Then
the assumption consists with the fact that the 
average along to $y$ direction is finite.
In other words, we realize 
the situation that the the second property of the conformal structure
holds for every $z \in \xi(p)^c$ 
and in the first property of the expectation value,
$\displaystyle{\Exp{\lambda(p,\lambda)}
\left(\left|\frac{\partial \uC{}}{\partial y}\right| \right)}$
is finite but
$\displaystyle{\Exp{\lambda(p,\lambda)}
\left(\left|\frac{\partial \uS{}}{\partial x}\right| \right)}$
diverges for $p \nearrow p_c(\xi, n)$ and $n\to \infty$.

It means that the family of $\{\COS{v}(\xi(p))\}_v$ gather in 
the very narrow slit
along a path from $\partial_\rght \Box_{nL}$ to
$\partial_\lft \Box_{nL}$. The density of the curves
at the cross section should
be parameterized by $(p_c(\xi, n)-p)^b$ of $b>0$.
Due to the ergodic properties, the position of the slits
and the cluster are governed by a point process
associated with the Poisson point process.
(It reminds us of the continuum fractal percolation \cite{MR}.)
Since the equipotential curves $\COS{v}(\xi(p))$ do not penetrate
into $\xi(p)$, they must go along the front of these clusters,
whereas the front is expected to have a fractal structure \cite{G}.
 
If the slit $d$ is bounded from below $d>\epsilon$, 
from Lemma \ref{lmm:confmap} (1),
$\displaystyle{
\left|\frac{\partial z}{\partial u}\right|}$ does not vanish and thus 
the arclength $\COS{v}(\xi(p))$ has the infinite length
for the limit $p \nearrow p_c(\xi,n)$ and $\gamma_0 \to 0$ under
the assumption.

When we consider the small distance and the fractal
structure, we also consider the 
limit of the radius $r \to 0$ by keeping the volume fraction
$p$, which is relatively
the same as large $n$-limit of $\Box_{nL}$ fixing $r$. 
Hence we can go on to consider the finite $r$ for a while.

We are concerned with the probability of existence
of the slits which has the gap $d < \epsilon$ for given $\epsilon$
(parameterized by finite $r$).
More precisely, we are interested in the positions in which 
the intensity $|\nabla \uC{}|$ is greater than $1/\epsilon'$ of $\epsilon'>0$.
Since $\GammaC{}(p)$ is finite,
we may have the fact
$\lim_{\epsilon' \to 0}\Exp{\lambda_c}(\ell(\{\ z \in \CC\ |
\ |\nabla \uC{}(z)|>1/\epsilon' \})) =0$ for $n\to \infty$
even for the limit $p \nearrow p_c(\xi,n)$.
Since the essentials preserve even for finite $n$,
it means that the length of 
$\COS{v}(\xi(p))$ may basically diverge
for sufficiently large $n$ under the limit $p \nearrow p_c(\xi,n)$.
The situation that curves with infinite length are embedded into
a finite region such as $B_{nL}$ might be related to the fractal
structure.

As we mentioned numerically in Ref.~\cite{MSW3},
based on the above arguments,
we conjecture the following propositions for $n \to \infty$:
\begin{enumerate}
\item
The equipotential curve $\COS{v}(\xi(p)) \subset 
\overline{\xi(p)^c}$ ($v \in (0, \uC{0})$)
has the fractal structure 
$\Prob{\lambda(p)}$-a.s. for the limit 
$p \nearrow p_c$, $r \to 0$, and $\gamma_0 \to 0$, and

\item
the equipotential curve $\CVS{u}(\xi(p)) \subset \xi(p)$ ($u \in (0, \uS{0})$)
has the fractal structure 
$\Prob{\lambda(p)}$-a.s. for the limit 
$p \searrow p_c^*$, $r \to 0$, and $\gamma_0 \to 0$.
\end{enumerate}
}}
\end{remark}

\bigskip

\begin{remark} \label{rmk:flow}
{\rm{
Instead of the large limit in above arguments, we
will consider another limit.
For $X \in \cQ_<(B_{nL})$,
by extending 
$I^{\uC{0}}_{f_0^+(X)}$
to $I^{m\uC{0}}_{mf_0^+(p)}$ 
we also have a conformal map $\psi^X_{m,n}:\Box_{nL}\to 
I^{m\uC{0}}_{mf_0^+(p)}$. 
By taking the limit $n,m \to \infty$
and compactification of them appropriately, 
the limit of the map $\psi^X_{\infty,\infty}$ 
could be regarded as a map between Riemann spheres $\PP$.
Then
$\psi^X_{\infty,\infty}$ has the properties that
its degree of the cover is 1 and each ramification index 
at each slit in  $I^{m\uC{0}}_{mf_0^+(p)}$ is also 1. 
For the case of finite points $\hXconf\in \hcQ(\PP)$, it is related to
the cylindrical flow $\psi_\flow^{(N)}$ 
problem to dilute $N$ cylinders case;
\begin{equation}
   \psi_\flow^{(N)} = 
\frac{u_0}{L}
\left( 
z + \sum_{i=0}^N\frac{\rho^2}{z - z_i}
\right)
\end{equation}
for $\rho \ll \min_{i,j, (i\neq j)}|z_i - z_j|$
illustrated in Figure \ref{fig:ManyClyndrs}.
\begin{figure}[h]
\begin{center}
\includegraphics[width=6cm]{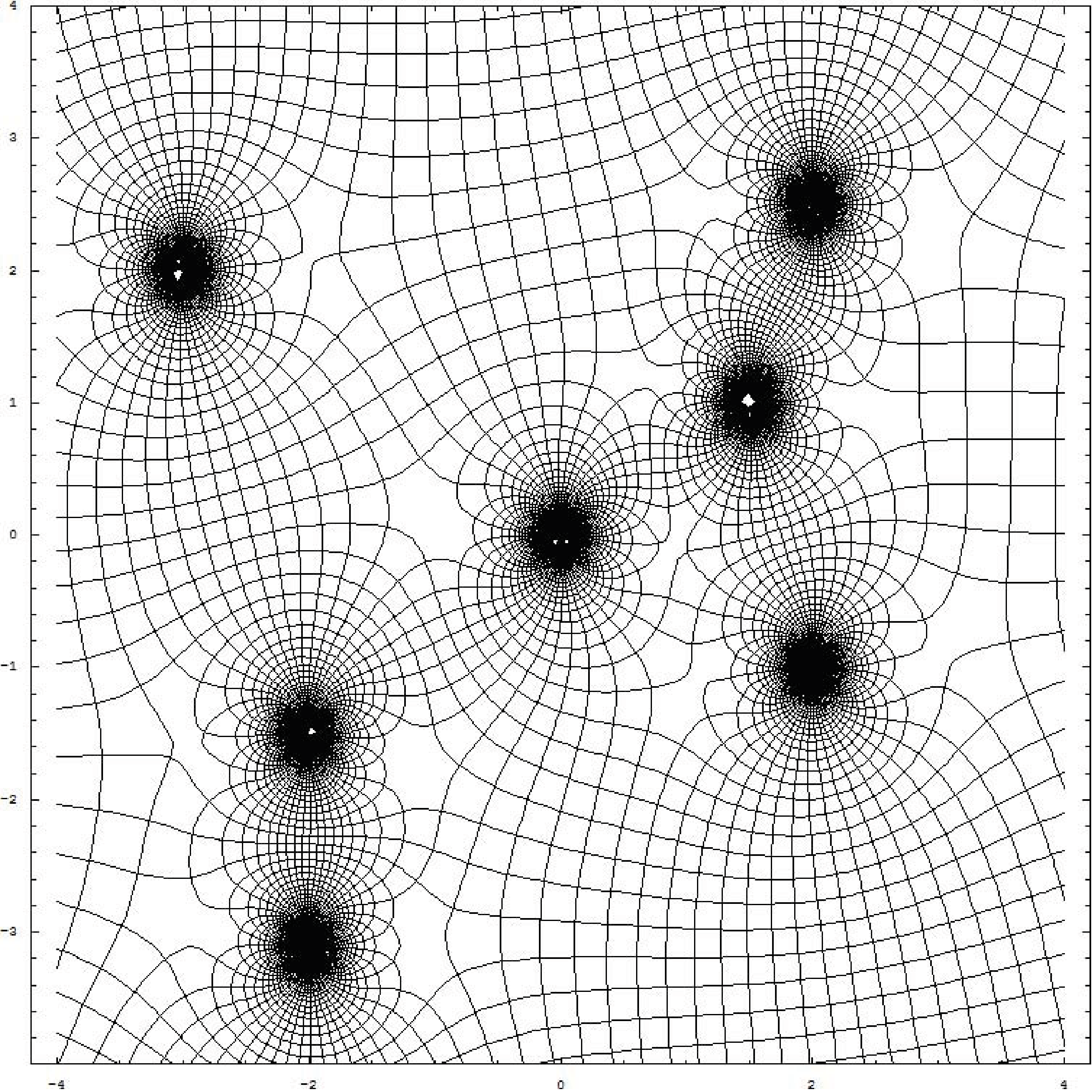}
\end{center}
\caption{A cylindrical flow $\psi_\flow^{(N)}$} 
\label{fig:ManyClyndrs}
\end{figure}

Since it is stated from the viewpoint of the conformal field theory \cite{BBCF}
that two-dimensional turbulence has
the fractal dimension of 4/3 
whereas the fractal dimension of the equipotential
curve in our system is also conjectured \cite{MSW3} to be 4/3,
the two-dimensional flow problem and our model might
be connected. 

Further the fact that
the volume fraction $p$ parameterizes the deformation
of $I_{u^-_0}^{u^+_0}$, $[0,1]^2 \to [0,1]\times[0,\infty]$
recalls the fundamental region of the elliptic module $\tau$
and then the duality might correspond to the Jacobi transformation
$\tau\to 1/\tau$.
}}
\end{remark}

\bigskip

\bigskip

\section{Numerical Computations of Conductivity in OS- and VS-types}

As we did in Ref.~\cite{MSW3},
we numerically solved the generalized Laplace equations 
(\ref{eq:0-1}) with (\ref{eq:0-3})
for $\Xconf_{p, i_s}$ which is parameterized by the volume fraction $p$
and pseud random seed $i_s$.
By monitoring the volume fraction $\Vol_{B_L}(\Xconf_n)$, 
we set the particles in $B_L$ one by one 
as long as $\Vol_{B_L}(\Xconf_n) \le p$ for the given volume fraction $p$.
We found the step $n(p)$ such that
$\Vol_{B_L}(\Xconf_{n(p)-1}) \le p$ and
$\Vol_{B_L}(\Xconf_{n(p)}) > p$.
Since for sufficiently large size of box, the difference 
$\delta \Vol_{B_L}(\Xconf_{n(p)-1})$ $:=\Vol_{B_L}(\Xconf_{n(p)})$
$-\Vol_{B_L}(\Xconf_{n(p)-1})$ is sufficiently small,
we regard $\Vol_{B_L}(\Xconf_{n(p)})$ as each $p$ hereafter
under this accuracy.

Since we used the pseudo-randomness to simulate the random
configuration $\Xconf_{n(p)}$ for given $p$,
the configuration $\Xconf_{n(p)}$ depends upon the
 seed $i_s$ of the pseudo-random which we choose and thus
let it be denoted by $\Xconf_{p, i_s}$.
For the same seed $i_s$ of the pseudo-random,
a configuration $\Xconf_{p,i_s}$ of a volume fraction $p$ 
naturally contains a configuration $\Xconf_{p', i_s}$ of $p'<p$
due to our algorithm, i.e., 
$$
\Xconf_{p',i_s} \subset \Xconf_{p, i_s}.
$$
In other words, we handle a path $\xi_{i_s}$ parameterized $i_s$
in $\cQ_r(\Box_{L})$ as in $\Path\cQ_r(\CC)$.
$\xi_{i_s}:[0,1] \to \cQ_r(\Box_L)$ such that $\Vol_{B_L}(\xi(p))=p$,
$\xi(p)\in \cQ_{r,\lambda(p)}(\Box_L)$ and
$\xi_{i_s}(p) \subset \xi_{i_s}(p')$ for every $p\le p'$.

Hence the elements in the set of the configurations 
$\{\Xconf_{p,i_s}\ | \ p \in [0,1]\}$ with
the same seed $i_s$ are relevant.
Figure \ref{fig:sigma_config} illustrates the configurations  
of the seed $i_s=1$ for several $p$'s.
\begin{figure}[htbp]
\begin{center}
\includegraphics[width=12cm]{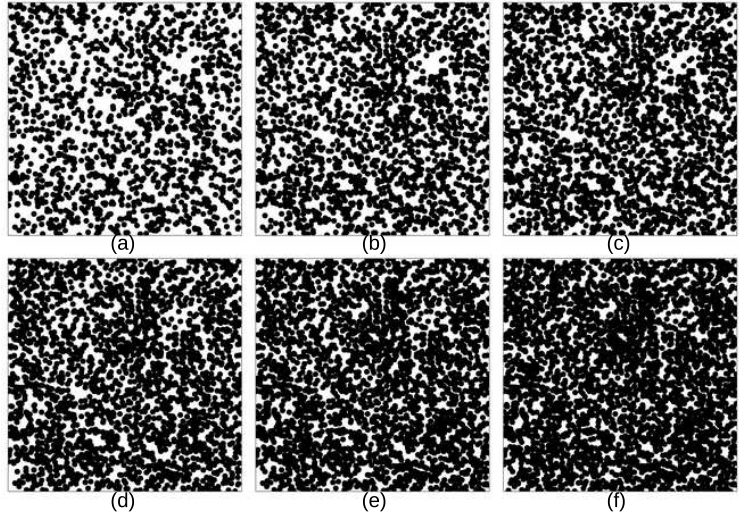}
\end{center}
\caption{The configurations of the seed $i_s=1$ with
the volume fractions 0.5, 0.6, 0.65, 0.7, 0.75, 0.8 for
(a), (b), (c), (d), (e), (f) respectively. }
\label{fig:sigma_config}
\end{figure}

\subsection{Numerical Parameters of Computations}

In the finite difference method computaions,
we used three lattices, $1024 \times 1024$, 
$2048 \times 2048$, and $4096 \times 4096$, to represent 
$\Box_L$, $\Box_{2L}$ and $\Box_{3L}$ respectively.
Since the radius of the particle $r$ corresponds to $25$ meshes,
we regard that $L = 40.96r$.
Hence the difference $\delta \Vol_{B_L}(\Xconf_{n(p)-1})$ 
is determined at most 
$1.87\times 10^{-3}$ for $\Box_L$,
$4.68\times 10^{-4}$ for $\Box_{2L}$, and
$1.17\times 10^{-4}$ for $\Box_{3L}$.
We set the infinitesimal conductivity $\cond_{\inf}=10^{-4}$ 
and the conductive one $\cond_\mat = 1$ as in Ref.~\cite{MSW3}.

\bigskip

\subsection{Results of Numerical Computations}

\subsubsection{Conductivity Curve}

We used the finite difference method to 
find the solutions of (\ref{eq:0-1}) by
handling the binary conductivity distribution $\cond(x,y)$,
as in the previous work \cite{MSW3}.

\begin{figure}[htbp]
\begin{minipage}{0.45\hsize}
 \begin{center}
 \includegraphics[height=7cm, angle=270]{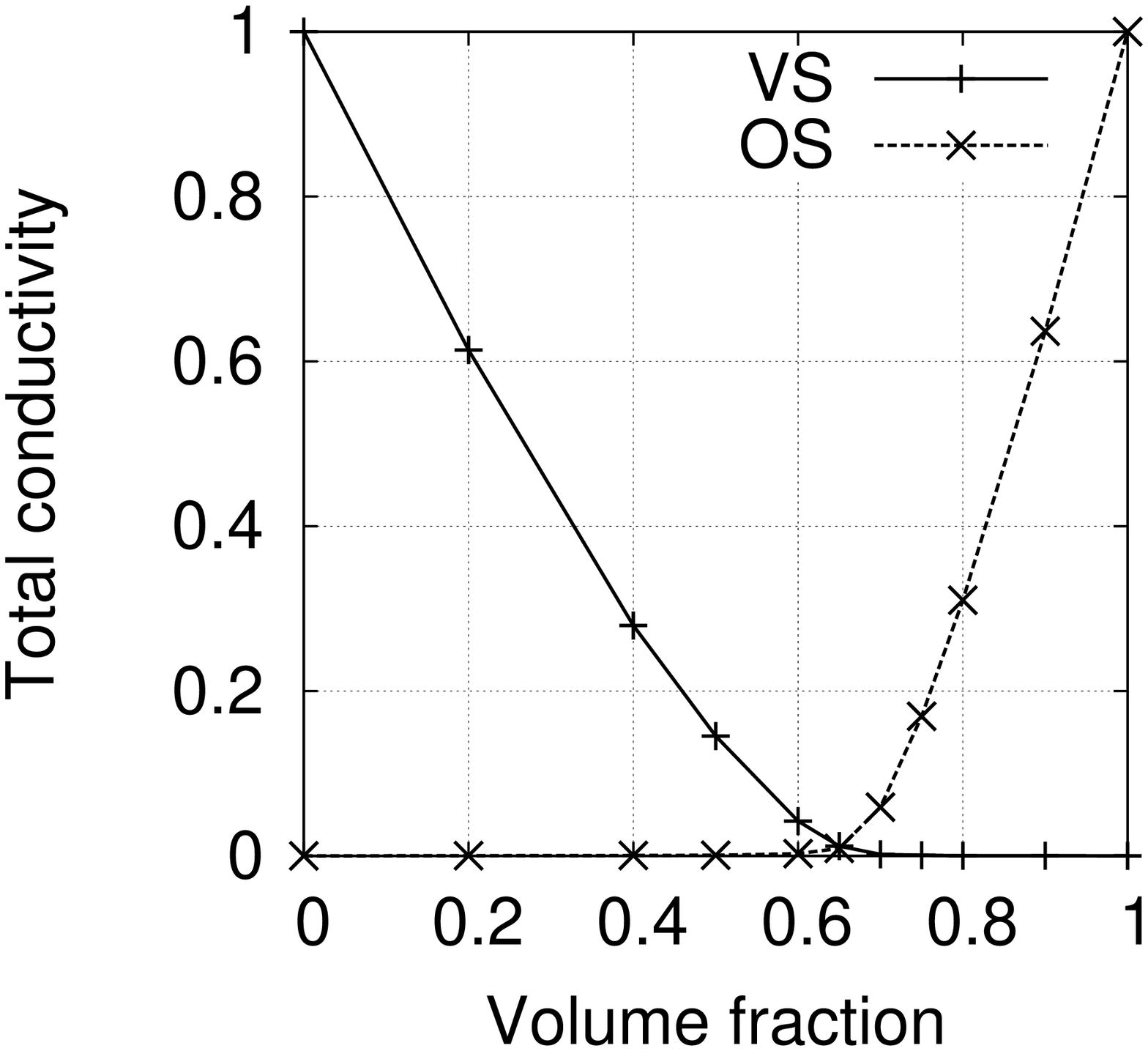} \newline
 (a)
 \end{center}
\end{minipage}
\begin{minipage}{0.45\hsize}
 \begin{center}
 \includegraphics[height=7cm, angle=270]{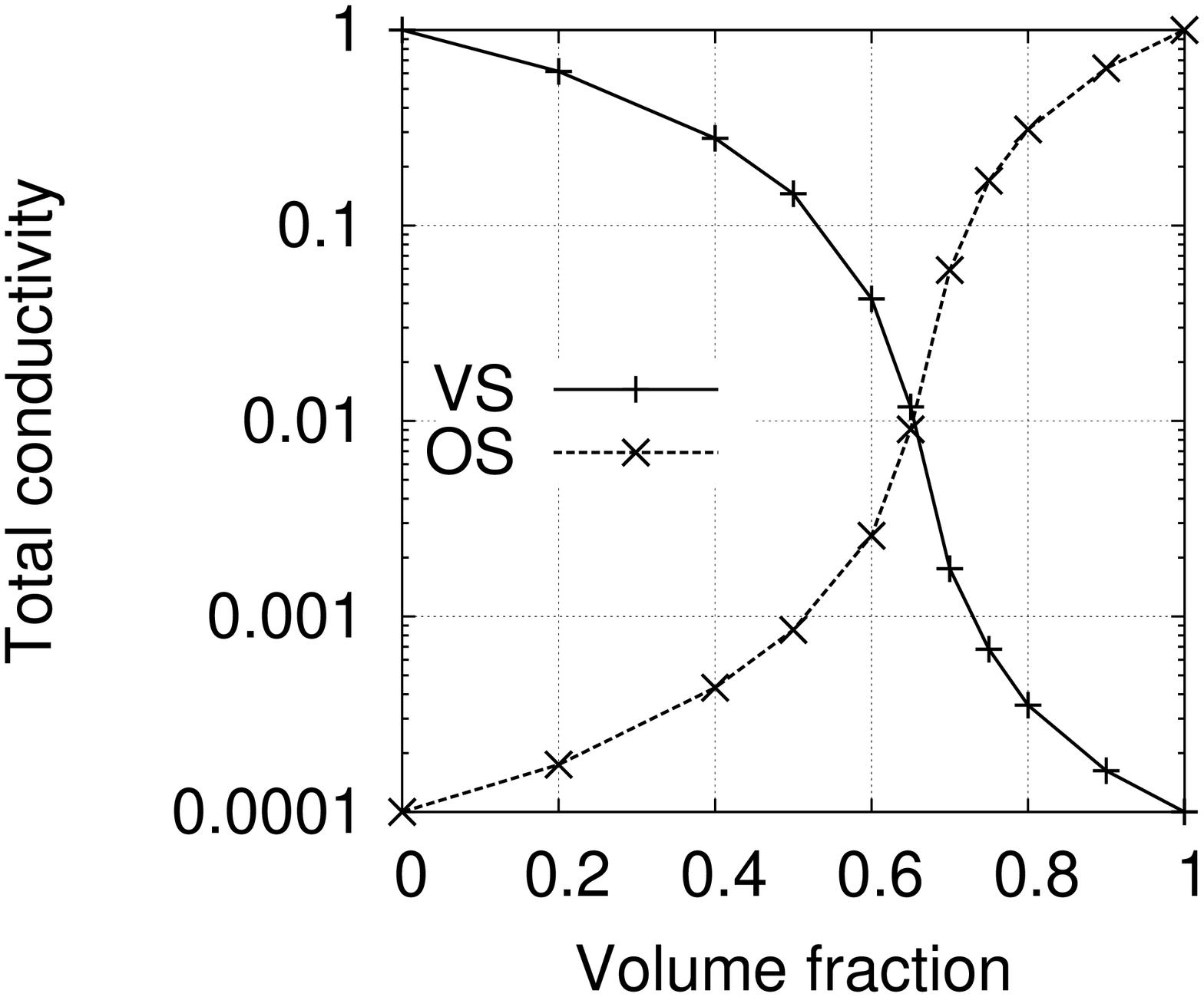} \newline
 (b)
 \end{center}
\end{minipage}
\caption{
Average of
conductivity curves for $\Box_{3L}$:
(a) for the linear total conductivity 
and (b) for its logarithm scale.
The lines are the linear interpolation of the computational
results.
}
\label{fig:conductivity_curve} 
\end{figure}

Figure \ref{fig:conductivity_curve}(a) exhibits the linear scale 
behavior of 
average of the total conductivities $\GammaCS{B_{3L}}(X)$ 
in $\Box_{3L}$
whereas Figure \ref{fig:conductivity_curve}(b) shows its
logarithm property. Even though Figure \ref{fig:conductivity_curve}(b) 
illustrates the property of the binary materials, the linear scale
behavior of the conductivity curves are described well 
by 
\begin{eqnarray}
\GammaC{>}(p) &= 
\dfrac{(p - \pcC)^\talphaC}{(1 - \pc)^\talphaC}
\vartheta(p-\pcC), \nonumber\\
\quad
\GammaS{<}(p) &= 
\dfrac{(\pcS-p)^\talhpaS}{(\pcS)^\talhpaS}
\vartheta(\pcS-p),
\label{eq:condtotalLin}
\end{eqnarray}
where $\pcC$ and $\pcS$ are the thresholds, $\talphaC$ and
$\talhpaS$ are the critical exponents, or merely called exponents,
and $\vartheta$ is a Heaviside function, i.e.,
$\vartheta(x) = 1$ if $x \ge 0$ and vanishes otherwise.

We computed six paths $\xi_{i_s}$ with different seeds $i_s$ of the 
pseudo-randomness for each $\Box_{aL}$, $(a=1,2,3)$.
Using the results we determined the thresholds $\pc$ and 
the exponents $\talphaC$ for each path
using the least mean square method
{\it{ in the linear scale resolution}}
respectively.
The obtained results are shown in Table \ref{table:lattice-dep}.
Let the fluctuations mean the interval of the minimal and maximal values.
It shows that
the dependence of the threshold and the exponent on the size of boxes.
The larger size of boxes is, 
the smaller the fluctuations of the threshold and the exponent are.
\begin{table}[htbp]
\caption{The Size dependence of the threshold and the exponent:}
\label{table:lattice-dep}
\begin{center}
\begin{tabular}{|c|c|c|c|c|c|c|c|c|c|c|c|}
\hline
\multicolumn{2}{|c|}{ } &
\multicolumn{2}{|c|}{Threshold}&
\multicolumn{2}{|c|}{Exponent} \\
\hline
\multicolumn{2}{|c|}{ } 
      &Average & Max-Min & Average & Max-Min  \\
\hline
   &$\Box_{L}$ & 0.669  & 0.097   & 1.229   & 0.690 \\
OS-type&$\Box_{2L}$ & 0.664  & 0.059   & 1.297   & 0.312 \\
 &$\Box_{3L}$ & 0.661  & 0.032   & 1.312   & 0.190 \\
\hline
    &$\Box_L$ & 0.677  & 0.115   & 1.462   & 0.433 \\
VS-type&$\Box_{2L}$ & 0.676  & 0.028   & 1.433   & 0.087 \\
    &$\Box_{3L}$ & 0.669  & 0.042   & 1.390   & 0.125 \\
\hline
\end{tabular}
\end{center}
\end{table}

\subsubsection{Fractal Dimension of Equipotential Curves}

We computed the fractal dimensions of these equipotential curves
which are shown in Figure \ref{fig:fractaldim} as in Ref.~\cite{MSW3};
we computed only OS-type there.
On the computation of the fractal dimensions,
we used the box-counting method \cite{MB2}.
Since the curves in Figure \ref{fig:contours} of the seed $i_s=1$
seem to have the fractal structure, 
their numerical evaluations of the fractal dimensions in Figure 
\ref{fig:fractaldim}
show that both equipotential curves of OS- and VS-types might have
non-trivial fractal dimension.

\begin{figure}[htbp]
\begin{minipage}{0.45\hsize}
 \begin{center}
  \includegraphics[height=7cm, angle=270]{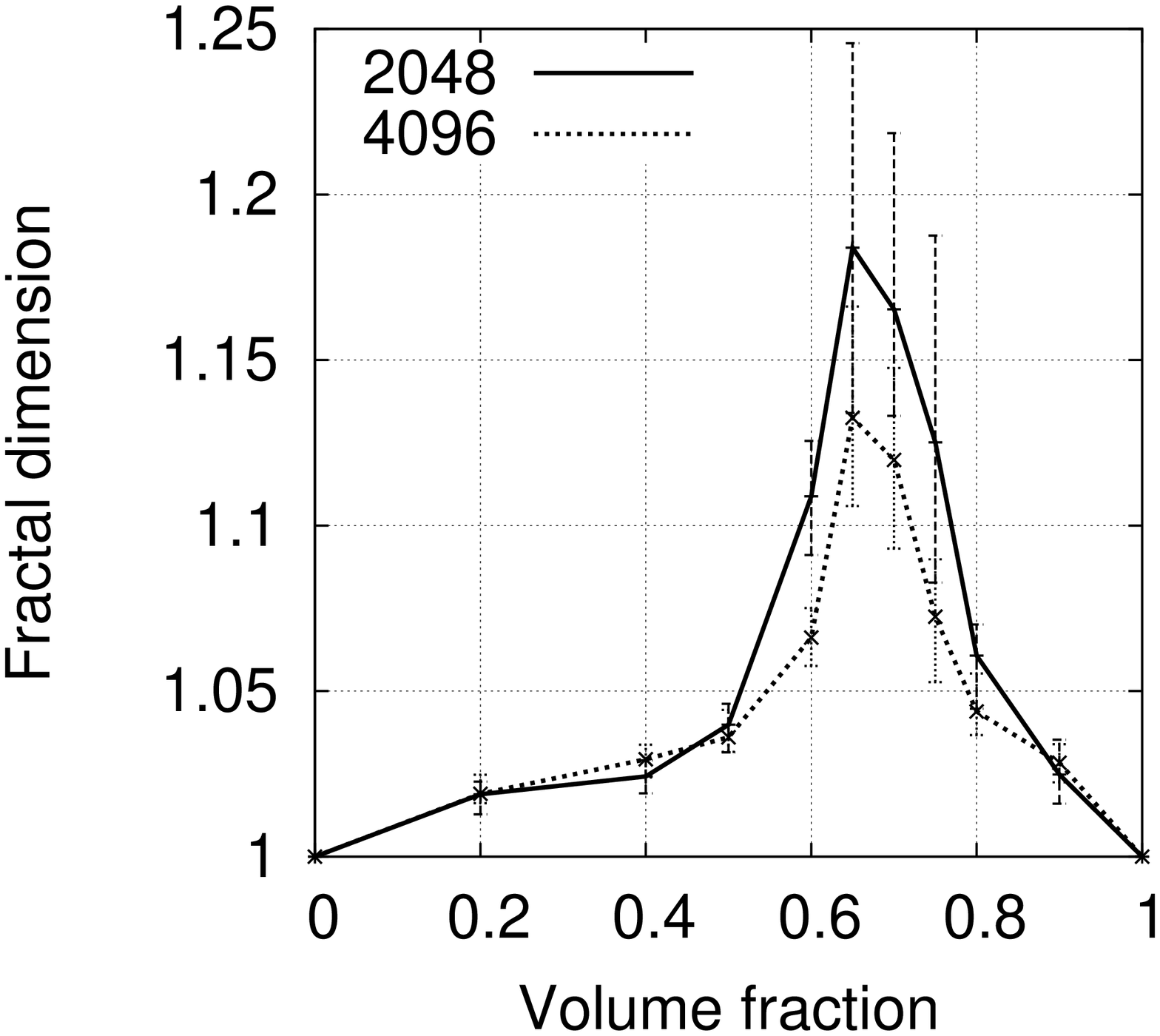} \newline
  (a)
 \end{center}
\end{minipage}
\begin{minipage}{0.45\hsize}
 \begin{center}
  \includegraphics[height=7cm, angle=270]{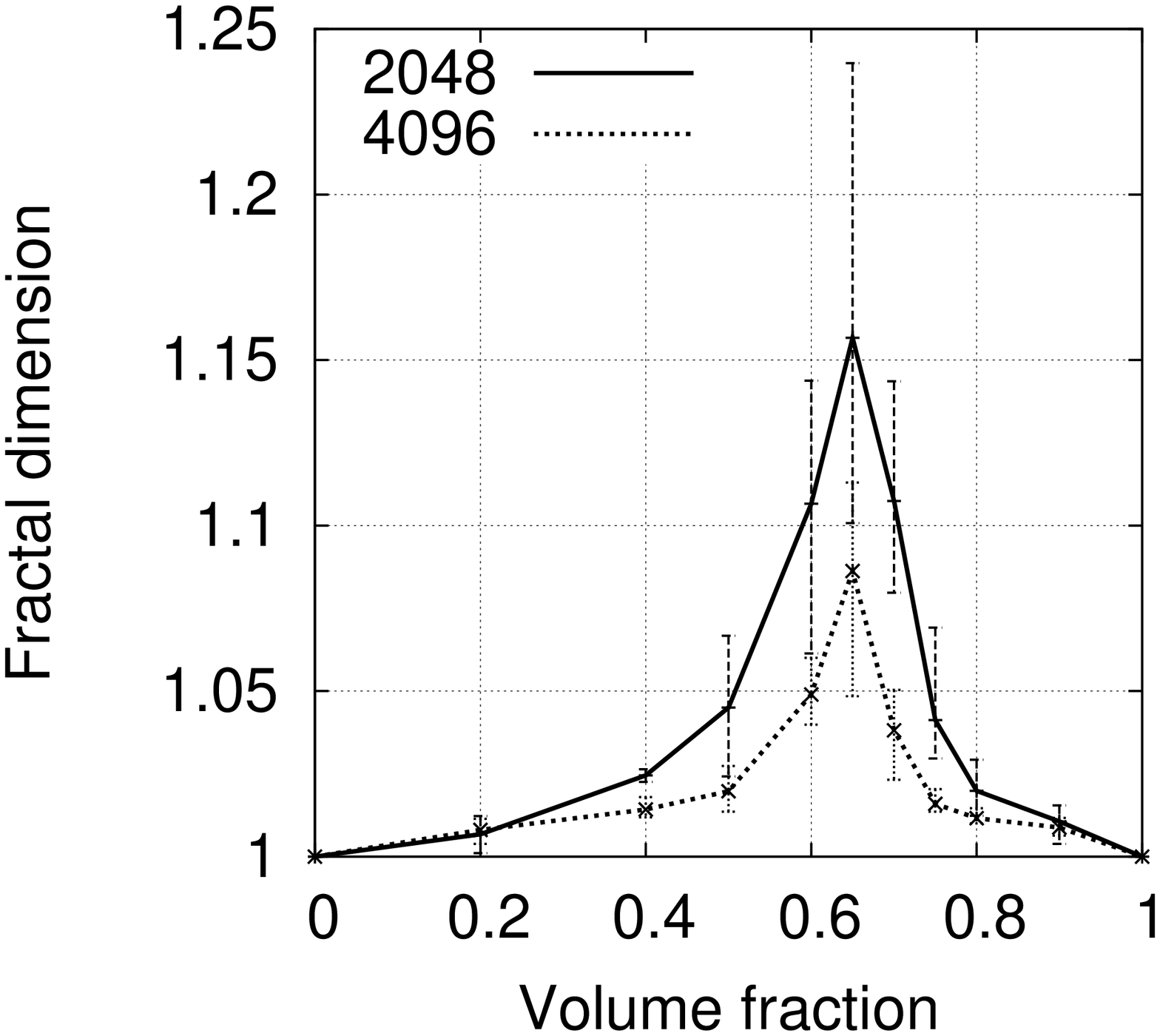} \newline
  (b)
 \end{center}
\end{minipage}
\caption{The fractal dimension vs. volume fraction $p$ for six seeds:
(a) for OS-type, (b) for VS-type.
}
\label{fig:fractaldim}
\end{figure}

\section{Approximation Formulae of Conductivity Curves}

In this section, we consider the conductivity curves
$\GammaCS{}(p)$
over all of the volume fraction $[0,1]$.
Using the Keller-Dykhne reciprocity law \cite{Dyk,Kell}
in Proposition \ref{rmk:gamma-conjugate},
we consider the following formulae $\GammaCS{0}(p)$:

\begin{lemma}\label{lemma:R_cond2}
For sufficiently small $\cond_\inf>0$ and $L\to \infty$,
\begin{eqnarray}
\GammaC{0}(p)&:=& 
\cond_\mat \dfrac{(p - \pc)^\talpha}{(1 - \pc)^\talpha}
\vartheta(p-\pc) + 
\dfrac{(\pc)^\talpha \sqrt{\cond_\mat \cond_\inf}}
{(\pc-p)^\talpha\sqrt{\cond_\mat/\cond_\inf} \vartheta(\pc-p)
+(\pc)^\talpha} 
, \nonumber\\
\GammaS{0}(p)&:=&
\cond_\mat \dfrac{(\pc - p)^\talpha}{(\pc)^\talpha}
\vartheta(\pc-p) + 
\dfrac{(1-\pc)^\talpha \sqrt{\cond_\mat \cond_\inf}}
{(p-\pc)^\talpha\sqrt{\cond_\mat/\cond_\inf} \vartheta(p-\pc)
+(1-\pc)^\talpha}, 
\label{eq:condtotalLog}
\end{eqnarray}
the following relations hold:

\begin{enumerate}
\item  $\GammaC{0}(p) \GammaS{0}(p) = \gamma_0 \gamma_1$.

\item  
$\displaystyle{
\GammaC{0}(p) = \sqrt{\gamma_0\gamma_1} 
  \left(1+
     \sqrt{\frac{\gamma_1}{\gamma_0}}\frac{(p-p_c)}{p_c}
+\cO((p-p_c)^2)\right)
}$
for $p \in (p_c(1-\sqrt[2t]{\gamma_1/\gamma_0}),$ $
p_c(1+\sqrt[2t]{\gamma_1/\gamma_0}))$.

\item 
$\displaystyle{
\GammaS{0}(p) = \sqrt{\gamma_0\gamma_1} 
  \left(1-
     \sqrt{\frac{\gamma_1}{\gamma_0}}\frac{(p - p_c)}{p_c}
+\cO((p-p_c)^2)\right)
}$
for $p \in (p_c(1-\sqrt[2t]{\gamma_1/\gamma_0}),$ $
p_c(1+\sqrt[2t]{\gamma_1/\gamma_0}))$.
\end{enumerate}
\end{lemma}

\begin{proof}
Noting $\vartheta(p-\pc)\vartheta(\pc-p) \equiv 0$,
$\dfrac{1}{\GammaS{0}(p)}$ is equal to 
{\small{
\begin{eqnarray}
&\ &
\frac{(p-\pc)^\talpha\sqrt{\frac{\cond_\mat}{\cond_\inf}} \vartheta(p-\pc)
+(1-\pc)^\talpha}{ 
\cond_\mat \frac{(\pc - p)^\talpha}{(\pc)^\talpha}
\vartheta(\pc-p)
\left((p-\pc)^\talpha\sqrt{\frac{\cond_\mat}{\cond_\inf}} \vartheta(p-\pc)
+(1-\pc)^\talpha\right) 
+(1-\pc)^\talpha \sqrt{\cond_\mat \cond_\inf}}
 \nonumber\\
 &=& 
\frac{\pc^t\left(
(p-\pc)^\talpha\sqrt{\frac{\cond_\mat}{\cond_\inf}} \vartheta(p-\pc)
+(1-\pc)^\talpha\right)
}{ 
\left(\sqrt{\frac{\cond_\mat}{\cond_\inf}} (\pc - p)^\talpha
\vartheta(\pc-p) +\pc^\talpha\right) 
(1-\pc)^\talpha \sqrt{\cond_\mat \cond_\inf}}
 \nonumber\\
 &=& 
\left(
\frac{(p-\pc)^\talpha}{(1-\pc)^\talpha}
\sqrt{\frac{\cond_\mat}{\cond_\inf}} \vartheta(p-\pc)
+1 \right)
\frac{\pc^\talpha}{ 
(\sqrt{\frac{\cond_\mat}{\cond_\inf}} (\pc - p)^\talpha
\vartheta(\pc-p)
+\pc^\talpha) \sqrt{\cond_\mat \cond_\inf}}
 \nonumber\\
 &=& 
\left(
\frac{(p-\pc)^\talpha}{(1-\pc)^\talpha}
\sqrt{\frac{\cond_\mat}{\cond_\inf}} \vartheta(p-\pc)
\right)
\frac{
(\sqrt{\frac{\cond_\mat}{\cond_\inf}} (\pc - p)^\talpha
\vartheta(\pc-p)
+\pc^\talpha) 
}{ 
(\sqrt{\frac{\cond_\mat}{\cond_\inf}} (\pc - p)^\talpha
\vartheta(\pc-p)
+\pc^\talpha) \sqrt{\cond_\mat \cond_\inf}
}
 \nonumber\\
 &+& 
\frac{
\pc^\talpha
}{ 
(\sqrt{\frac{\cond_\mat}{\cond_\inf}} (\pc - p)^\talpha
\vartheta(\pc-p)
+\pc^\talpha) \sqrt{\cond_\mat \cond_\inf}
}
 \nonumber\\
 &=& 
\frac{1}{\cond_\inf\cond_\mat}
\left(
\cond_\mat
\frac{(p-\pc)^\talpha}{(1-\pc)^\talpha}
\vartheta(p-\pc)
+
\frac{\pc^\talpha \sqrt{\cond_\mat \cond_\inf} }{ 
\sqrt{\frac{\cond_\mat}{\cond_\inf}} (\pc - p)^\talpha
\vartheta(\pc-p)
+\pc^\talpha } \right)
=\dfrac{1}{\cond_\inf\cond_\mat}\GammaC{0}(p).
 \nonumber
\end{eqnarray}
}}

\noindent
For $(p_c(1-\sqrt[2t]{\gamma_1/\gamma_0}), p_c)$, we expand 
$\GammaC{0}(p)$ by
$\displaystyle{
\sqrt{\gamma_0\gamma_1}}$  
$\displaystyle{
  \left(1-
     \sqrt{\frac{\gamma_1}{\gamma_0}}\frac{(p_c - p)}{p_c}+\cdots\right)
}$. 
Similarly we also have (3) for 
$(p_c, p_c(1+\sqrt[2t]{\gamma_1/\gamma_0}))$.
The relation (1) gives the (2) and (3).
\end{proof}

\begin{remark}\label{rmk:R_cond2}
{\rm{
Figure \ref{fig:coodinate_cond} shows that our formulae
 (\ref{eq:condtotalLog}) for $(t=1.451, p_c=0.656)$
exhibits well the numerical computational
results. It implies that the results are expressed by continuous curves.
We conjecture that $\GammaCS{}(p)$ is approximated well by
$\GammaCS{0}(p)$, whose behavior around $p_c$ are given by
the formulae in (2) and (3) of Lemma \ref{lemma:R_cond2}
for 
$(p_c(1-\sqrt[2t]{\gamma_1/\gamma_0}),p_c(1+\sqrt[2t]{\gamma_1/\gamma_0}))$.
}}
\end{remark}

\begin{figure}[h]
\begin{minipage}{0.45\hsize}
 \begin{center}
\includegraphics[width=5cm,angle=270]{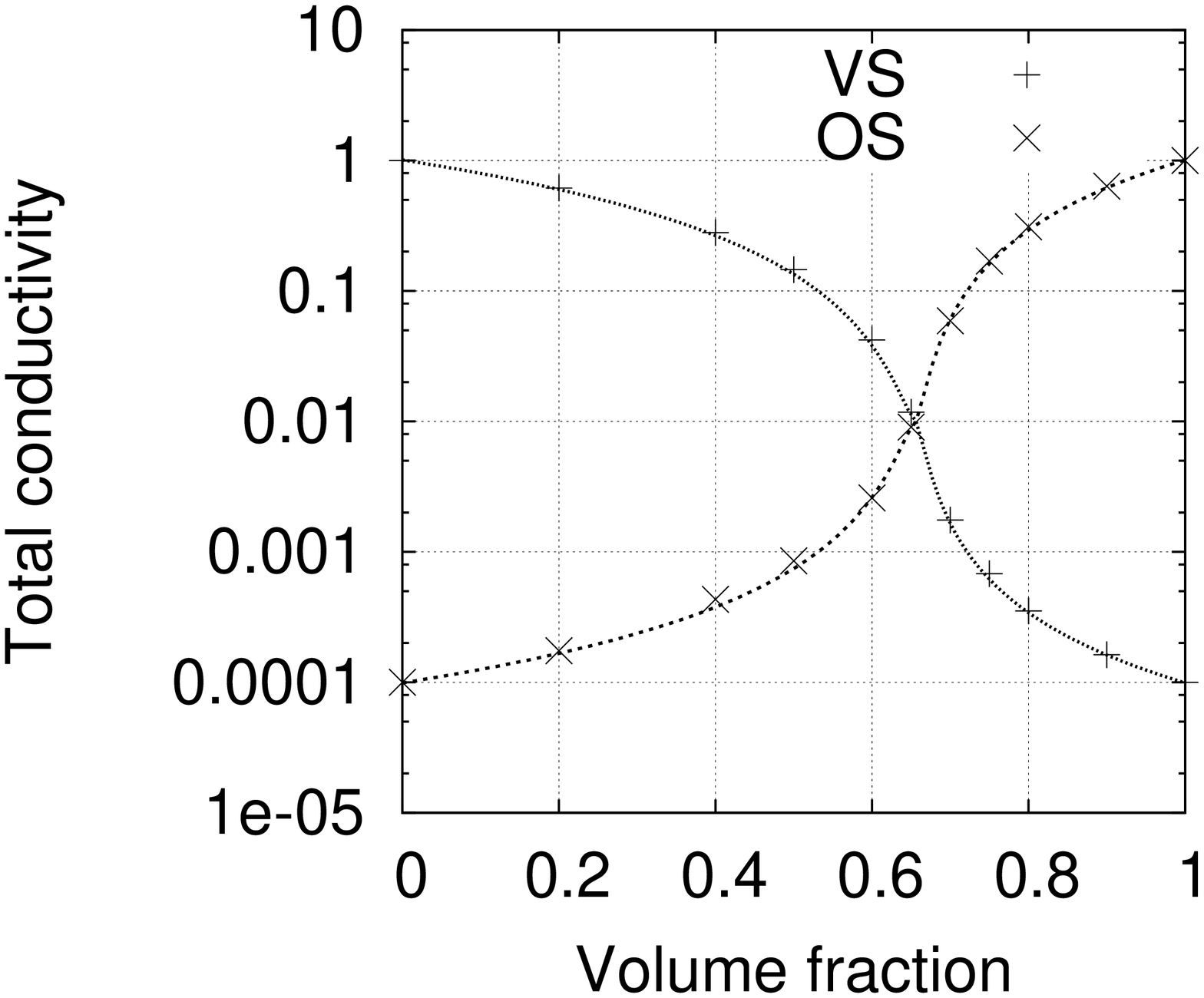}
 \newline
(a)
 \end{center}
\end{minipage}
\begin{minipage}{0.45\hsize}
 \begin{center}
\includegraphics[width=5cm,angle=270]{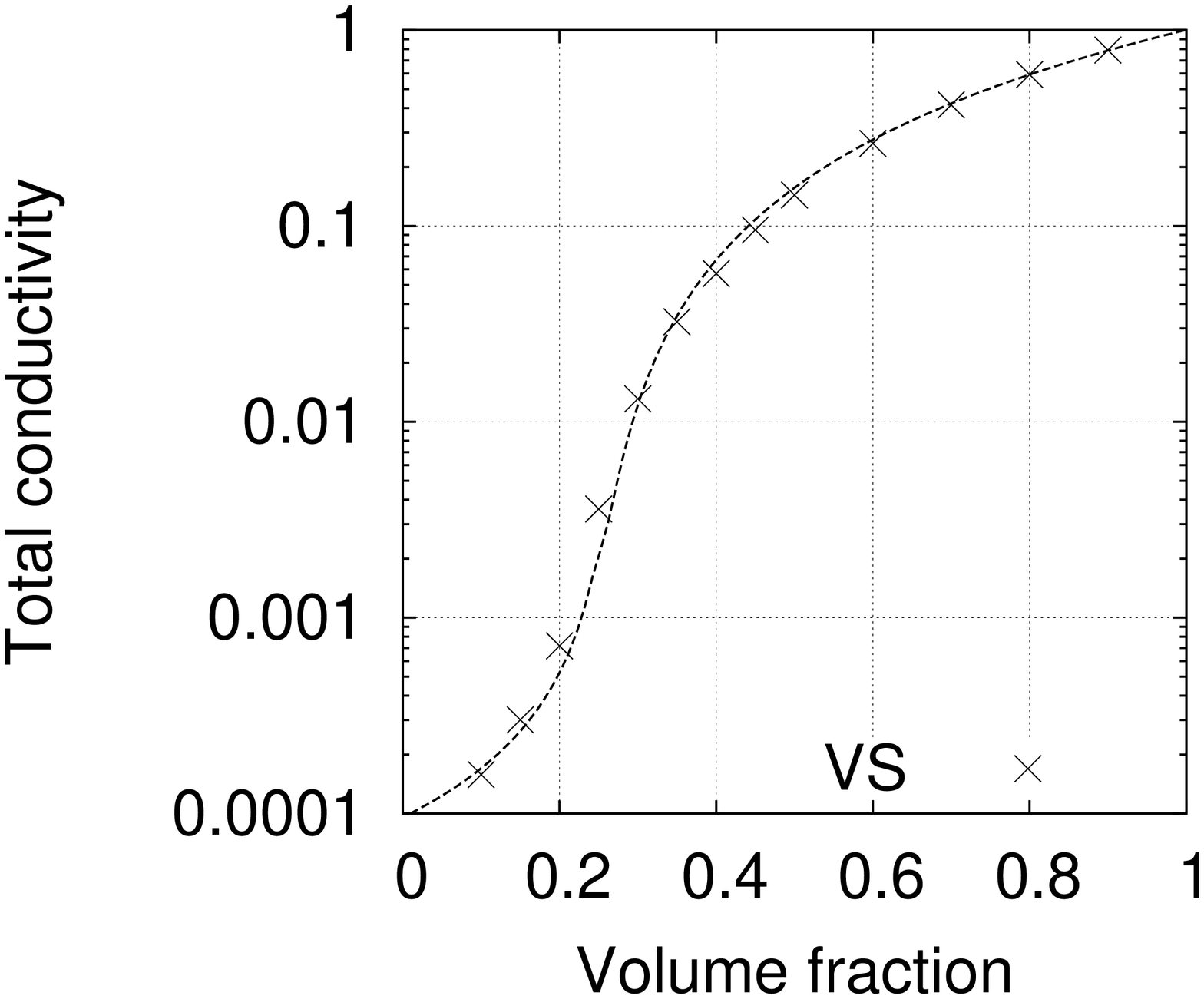}
\newline
 (b)
 \end{center}
\end{minipage}
\caption{The approximation formulae
and computed results:
(a) is the two dimensional case
and (b) the three dimensional case.
The curves are given by the formulae (6.1) and (6.2) and 
$+$ and $\times$ are the average of numerical computaions for
six seeds.
}
\label{fig:coodinate_cond}
\end{figure}

\begin{remark}\label{rmk:R_cond3}
{\rm{
It is expected that the formula (\ref{eq:condtotalLog}) 
might be generalized to three dimensional
case. 
As illustrated in Figure \ref{fig:coodinate_cond},
the computational results of  three dimensional CPM computed in 
Ref.~\cite{MSW2} may be written by 
\begin{equation}
\GammaC{\mathrm{3d}}(p) = 
\cond_\mat \frac{(p - \pc)^\talpha}{(1 - \pc)^\talpha}
\vartheta(p-\pc) + 
\frac{(\pc)^{t'} \sqrt[3]{\cond_\mat \cond_\inf^2}}
{(\pc-p)^{t'}\sqrt[3]{\cond_\mat/\cond_\inf} \vartheta(\pc-p)
+(\pc)^{t'}}, 
 \label{eq:condtotalLog3}
\end{equation}
where $t=1.7$, $t'= 1.2$, and $\pc = 0.25$.
Figure \ref{fig:coodinate_cond} (b) shows 
this formula also represents well the numerical computational results
in Ref.~\cite{MSW2}.

In Ref.~\cite{MSW}, we investigated the shape effects
in the continuum percoaltion conductivity problems and 
in Ref.~\cite{MSW2}, we studied the agglomeration effects on
the homogenized conductivity.
In terms of this novel formula, we can parameterize
these conductivity properties in 
Refs.~\cite{MSW} and \cite{MSW2} more precisely.
Further it might be also related to the elastic properties in
CPM as in \cite{WJ}.
}}
\end{remark}




\section*{Acknowledgment}
We thank Professors Tomoyuki Shirai, Yasuaki Hiraoka and Norio Konno
for valuable discussions.
We are also deeply grateful to Professor Makoto Katori
for his critical comments, especially on the duality
and the relation to a slit model,
and crucial discussions
and to Professor Yasuhide Fukumoto for 
valuable discussions and directing our attentions to
Ref.~\cite{AIM} and its related topics.



\begin{thebibliography}{00}
\bibitem{ADI} \by{N. Ansini, G. Dal Maso and C. Ida Zeppieri}
    \Paper{$\Gamma$-convergence and $H$-convergence of linear elliptic 
  operator}
    \jour{J. Math. Pures Appl.}
    \vol{99}
    \yr{2013}
    \pages{321-329}.

\bibitem{AIM} \by{K. Astala, T. Iwaniec and G. Martin}
    \Book{Elliptic partial differential equations and quasiconformal mappings 
     in the plane}
     \publ{Princeton}
   \byr{2008}.

\bibitem{BBCF}
    \by{D. Bernard, G. Boffetta, A. Celani and G. Falkovich}
    \Paper{Conformal invariance in two-dimensional turbulence}
    \jour{Nature Physics},
    \vol{2}
    \yr{2006}
    \pages{124-128}.


\bibitem{BMW} \by{A. Bourgeat, A. Mikeli\'c and S. Wright}
\Paper{Stochastic two-scale convergence in the mean and applications}
    \jour{J. die reine  angew. Math.}
    \vol{456}
    \yr{1994}
    \pages{19-51}.



\bibitem{Br} \by{H. Brezis}
    \Book{Functional analysis, Sobolev spaces and 
partial differential equations}
     \publ{Springer}
   \byr{2011}.

\bibitem{DM} \by{G. Dal Maso and L. Modica}
    \Paper{Nonlinear stochastic homogenization and ergodic}
    \jour{J. die reine  angew. Math.}
    \vol{326}
    \yr{1986}
    \pages{28-42}.


\bibitem{Dyk} \by{A. M. Dykhne}
    \Paper{Conductivity of a two-dimensional two-phase system}
    \jour{Sovit Phys. JETP}
    \vol{32}
    \yr{1971}
    \pages{63-65}.


\bibitem{GiTr} \by{D. Gilbarg and N. S. Trudinger}
    \Book{Elliptic partial differential equations of second order}
     \publ{Springer}
   \byr{2008}.


\bibitem{EDG} \by{E. De Giorgi}
    \Book{Ennio De Giorgi: selected papers}
     \publ{Springer}
   \byr{2006}.

\bibitem{G}
\by{G. Grimmest}
\Book{Percolation, second ed.}
 \publ{Springer}
\publaddr{Berlin}
 \byr{1999}.


\bibitem{J} \by{D. Jeulin}
    \Paper{Morphology and effective properties of multi-scale 
random sets: A review}
    \jour{C. R. Mecanique} 
    \vol{340}
    \yr{2012}
    \pages{219-229}.




\bibitem{Kell} \by{J. B. Keller}
    \Paper{A theorem on the conductivity of a composite medium}
    \jour{J.~Mah.~Phys.}
    \vol{5}
    \yr{1964}
    \pages{548-549}.

\bibitem{Kon} \by{D. Kontogiannis}
    \Paper{Homogenization problems in random media}
    \jour{Iowa State University, Thesis}
    \yr{2010}.


\bibitem{Koz} \by{S. M. Kozlov}
    \Paper{Geometric aspect of averaging}
    \jour{Russian.\ Math.\ Surveys}
    \vol{44}
    \yr{1989}
    \pages{91-144}.



\bibitem{MB2} 
\by{B. Mandelbrot}
    \Book{Fractal geometry of nature}
     \publ{Freeman}
   \byr{1982}.

\bibitem{MK} \by{V. A. Marcher and E. Y. Khrushchev}
    \Book{Homogenization of partial differential equations}
     \publ{Birkh\"auser}
   \byr{2004}.

\bibitem{MSW} \by{S. Matsutani and Y. Shimosako and Y. Wang}
    \Paper{Numerical computations of conductivity in continuum percolation 
             for overlapping spheroids}
    \jour{Int. J. Mod. Phys. C}
    \vol{21}
    \yr{2010}
    \pages{709-729}.

\bibitem{MSW2} \by{S. Matsutani and Y. Shimosako and Y. Wang}
    \Paper{Numerical computaions of conductivity over 
              agglomerated continuum percolation modes}
    \jour{Appl. Math. Modeling}
    \vol{37}
    \yr{2013}
    \pages{4007-4022}.


\bibitem{MSW3} \by{S. Matsutani and Y. Shimosako and Y. Wang}
    \Paper{Fractal structure of equipotential curves on a 
             continuum percolation model}
    \jour{Physica A}
    \vol{391}
    \yr{2012}
    \pages{5802-5809}.


\bibitem{MR} \by{R. Meester and R Roy}
    \Book{Continuum percolation}, Cambridge Tracts in Mathematics 119
     \publ{Cambridge}
   \byr{1996}.

\bibitem{MT1} \by{F. Murant and L. Tartar}
    \Paper{$H$-convergence} \pages{21-44}
    \Book{Topics in the mathematical modelling of composite materials}
    edited by A. Cherkaev and R. Kohn,
     \publ{Birkh\"auser}
   \byr{1997}.
    
\bibitem{MT2} \by{F. Murant and L. Tartar}
    \Paper{Calculus of variations and homogenization} \pages{139-174}
    \Book{Topics in the mathematical modelling of composite materials}
    edited by A. Cherkaev and R. Kohn,
     \publ{Birkh\"auser}
   \byr{1997}.
  


\bibitem{SKM} \by{D. Stoyan, W. S. Kendall and J. Mecke}
    \Book{Stochastic geometry and its applications} second ed.,
     \publ{Wiley}
   \byr{1995}.
  





\bibitem{Tr73} \by{N. S. Trudinger}
    \Paper{Linear elliptic operators with measurable coefficients}
    \jour{Annali Scuola Normale Superiore Pisa}
    \vol{27}
    \yr{1973}
    \pages{265-308}.

\bibitem{Tr77} \by{N. S. Trudinger}
    \Paper{Maximum principles for linear, non-uniformly elliptic 
         operators with measurable coefficients}
    \jour{Math. Z.}
    \vol{156}
    \yr{1977}
    \pages{291-301}.


\bibitem{WJ} \by{F. Willot and D. Jeulin}
    \Paper{Elastic and electrical behavior of some random 
multiscale highly-contrasted composites}
    \jour{Int. J.  Multiscale Comp. Eng.}
    \vol{9}
    \yr{2011}
    \pages{305-326}.



\end{thebibliography}
\end{document}